\documentclass[prd,twocolumn,nofootinbib,amsfonts,amssymb,amsmath]{revtex4-2}

\usepackage[utf8]{inputenc}

\usepackage{hyperref}
\usepackage{csquotes}
\usepackage{amsthm}
\usepackage{graphicx}

\newtheorem{theorem}{Theorem}

\newtheorem{corollary}[theorem]{Corollary}

\DeclareMathOperator{\Tr}{Tr}
\newcommand{\CombinationSet}[1]{\mathcal{C} \left(#1\right)}
\newcommand{\EvenCombinationSet}[1]{\bar{\mathcal{C}} \left(#1\right)}

\DeclareMathOperator{\Differential}{D}
\newcommand{\DifferentialAt}[1]{\left . \Differential \right \vert_{#1}}
\newcommand{\DifferentialNAt}[2]{\left . \Differential^{#1} \right \vert_{#2}}

\newcommand{\SymmetricGroup}{\mathrm{Sym}}
\newcommand{\SymmetricGroupN}[1]{\mathrm{Sym}_{#1}}

\begin{document}
\setlength{\abovedisplayskip}{1.5ex}
\setlength{\belowdisplayskip}{1.5ex}

\title{Existence and Construction of Exact FRG Flows of a UV-Interacting Scalar Field Theory}
\author{Jobst Ziebell}
\affiliation{Theoretisch-Physikalisches Institut, Abbe Center of Photonics, Friedrich-Schiller-Universit\"at Jena, Max-Wien-Platz 1, 07743 Jena, Germany}

\begin{abstract}
We prove the existence and give a construction procedure of Euclidean-invariant exact solutions to the Wetterich equation\cite{src:Wetterich1993} in $d > 2$ dimensions satisfying the naive boundary condition of a massive and interacting real scalar $\phi^4$ theory in the ultraviolet limit as well as a generalised free theory in the infrared 
limit.
The construction produces the momentum-dependent correlation functions to all orders through an iterative scheme, based on a self-consistent ansatz for the four-point function.
The resulting correlators are bounded at all regulator scales and we determine explicit bounds capturing the asymptotics in the UV and IR limits.
Furthermore, the given construction principle may be extended to other systems and might become useful in the study of general properties of exact solutions.
\end{abstract}

\maketitle

\section{Introduction}
In quantum field theory and related fields one rarely has access to exact expressions for quantities of interest.
Instead, one generally resorts to approximation schemes such as truncations of power series or lattice dicretisations.
But the use of such approximations raises the question of their respective reliability.
In terms of observables, one is interested in quantitative bounds on deviations from exact values.
However, the necessity of renormalisation turns the analysis of such deviations into a complicated task.
They are commonly studied by investigating artificial regulator dependencies, the apparent convergence of truncation schemes or by purely qualitative methods such as apparent stability of features like fixed points or phase transitions.
Nonetheless, it usually remains very difficult and often practically impossible to provide quantitative bounds on absolute errors and hence to explicitly specify the region of applicability of any given approximation procedure.

There are some notable exceptional cases in which exact results have been obtained such as the Schwinger model\cite{src:Schwinger1962}, the Thirring model\cite{src:Thirring1958} and the lattice $\phi^4_3$ and $\phi^4_{d>4}$ theories\cite{src:Brydges1983,src:Aizenman1981}.
Further exact results in quantum field theoretical models\cite{src:Rychkov2020}, condensed matter physics\cite{src:Schuetz2005} as well as in hydrodynamics\cite{src:Canet2016} and statistical mechanics\cite{src:Benitez2013} have been obtained through the use of the \enquote{functional renormalisation group} which is also at the core of this paper.
It constitutes a renormalisation scheme of the path integral quantisation and leads to well known expressions for the renormalisation group flow.
These include the Wegner-Houghton\cite{src:WegnerHoughton1973}, the Polchinski\cite{src:Polchinski1984} and the Wetterich\cite{src:Wetterich1993} equations, the latter being at the focus of this work.
In particular, it is also routinely used in studies of asymptotic safety scenarios of quantum gravity\cite{src:ReuterSaueressig2019,src:Percacci2017}.
For reviews and further applications, see \cite{src:EichhornEtAl2020,src:MetznerEtAl2012,src:Berges2002,src:Pawlowski2007,src:Gies2012,src:Delamotte2012,src:Braun2012,src:Nagy2014}.

The expansion of the Wetterich equation in powers of quantum fields corresponds to an expansion in one-particle irreducible vertices.
It constitutes a countably infinite tower of non-linear ordinary differential equations encoding the renormalisation group flow of the correlation functions of the quantum field theory at hand.
As will be demonstrated, it is possible to bootstrap formally (in the sense of not necessarily analytic) exact solutions to these equations by providing a well-behaved, consistent set of low-order correlation functions and giving an explicit construction procedure for the higher-order ones.
In this paper the above method is employed to construct exact solutions to the Wetterich equation for quantum field theories on Euclidean spacetimes of dimensions $d > 2$ that satisfy the naive boundary conditions of massive and interacting real scalar $\phi^4$ theories in the classical limit.
This boundary condition corresponds to strictly finite renormalisations of all coupling constants and hence does not agree with the rigorously known results for the $\phi^4_3$ theory.
In particular, the constructed solutions are shown to correspond to generalised free quantum field theories.

Nonetheless, I believe that exact solutions may provide good grounds for further research on the functional renormalisation group and its applications.
Through their constructive nature the solutions given in this paper may also be able to open the door to more rigorous error estimates because the knowledge of bounds on lower-order correlators may be employed to produce bounds on higher-order ones.
\section{The Functional Renormalisation Group}
Let us start with the Euclidean path integral quantisation of a classical action $S_\Lambda$ for a real scalar field at a UV regularisation scale $\Lambda > 0$.
Then
\begin{equation}
\begin{aligned}
\exp &\left[ -\Gamma \left( \phi \right) \right]
= \\
&\int \mathcal{D}_\Lambda \psi
\exp \left[ - S_\Lambda \left( \phi + \psi \right) + \left( \DifferentialAt{\phi} \Gamma \right) \left( \psi \right) \right] \, ,
\end{aligned}
\end{equation}
where $\mathcal{D}_\Lambda$ denotes the regularised path integral measure and $\Gamma$ is the effective action.
For clarity and brevity we shall use Fréchet derivatives instead of functional derivatives throughout this work which are related by
\begin{equation}
\left( \DifferentialAt{\phi} \Gamma \right) \left( \psi \right)
=
\int_{\mathbb{R}^d}
\frac{\delta \Gamma \left( \phi \right)}{\delta \phi \left( x \right)}
\psi \left( x \right)
\mathrm{d} x
\end{equation}
for all test functions $\psi$.
Introducing the effective average action $\Gamma_{k,\Lambda}$, one obtains\cite{src:ReuterWetterich1994,src:Manrique2009}
\begin{equation}
\label{eq:PathIntegralRepresentation}
\begin{aligned}
\exp \left[ -\Gamma_{k,\Lambda} \left( \phi \right) \right]
&=
\int \mathcal{D}_\Lambda \psi
\exp \Big[ - S_\Lambda \left( \phi + \psi \right) \\
&\phantom{=}
+ \left( \DifferentialAt{\phi} \Gamma_{k,\Lambda} \right) \left( \psi \right)
- \frac{1}{2} \left \langle \psi, R_k \psi \right \rangle
\Big] \, ,
\end{aligned}
\end{equation}
where $R_k$ is a suitable scale-dependent regulator and $\left \langle \cdot, \cdot \right \rangle$ denotes the standard inner product on $L^2 \left( \mathbb{R}^d \right)$.
In particular, for $k \to 0$ the regulator $R_k$ should vanish such that $\lim_{k \to 0} \Gamma_{k,\Lambda}$ reproduces the ordinary effective action $\Gamma$.
On the other extreme $R_k$ should diverge when $k \to \Lambda$ so that it acts as a delta functional with respect to the path integral ensuring $\lim_{k \to \Lambda} \Gamma_{k,\Lambda} \approx S_\Lambda$\cite{src:Wetterich1993} although it is known that this correspondence involves a reconstruction problem\cite{src:Manrique2009}.

Through the standard derivations one also obtains the Wetterich equation\cite{src:Wetterich1993,src:Ellwanger1994,src:Morris1994}
\begin{equation}
\label{eq:WetterichEquation}
\partial_k \Gamma_{k,\Lambda} \left( \phi \right) = \frac{1}{2} \Tr_\Lambda \left[ \left( \partial_k R_k \right) \left( \left . \Gamma_{k,\Lambda}^{\left(2\right)} \right \vert_{\phi} + R_k \right)^{-1} \right] \,,
\end{equation}
where $\left . \Gamma_{k,\Lambda}^{\left(2\right)} \right \vert_{\phi}$ denotes the second derivative of $\Gamma_{k,\Lambda}$ at $\phi$ interpreted as an operator\footnote{i.e for all test functions $\psi_1, \psi_2$ on $\mathbb{R}^d$ we have \[ \left \langle \psi_1^\ast, \left . \Gamma_k^{\left(2\right)} \right \vert_{\phi} \psi_2 \right \rangle = \int_{\mathbb{R}^d} \psi_1 \left . \Gamma_{k,\Lambda}^{\left(2\right)} \right \vert_{\phi} \psi_2 = \left( \DifferentialNAt{2}{\phi} \Gamma_k \right) \left( \psi_1, \psi_2 \right) \, . \]}.
For particularly well-behaved regulators one may now simply take the limit $\Lambda \to \infty$ in this equation, removing the necessity of a UV cutoff $\Lambda$ and leading to the \enquote{$\Lambda$-free}\cite{src:Manrique2009} form of equation \ref{eq:WetterichEquation}.
Let us refer to the resulting object of interest as $\Gamma_k = \lim_{\Lambda \to \infty} \Gamma_{k,\Lambda}$ to which we shall devote our attention throughout this paper.

Expanding the right hand side of this $\Lambda$-free equation in powers of a real scalar field $\phi$ gives us
\begin{equation}
\label{eq:WetterichExpandedFull}
\partial_k \Gamma_k \left( \phi \right) = \frac{1}{2} \sum_{n = 1}^\infty \frac{1}{n!} \Tr \left[ \left( \partial_k R_k \right) \left( \DifferentialNAt{n}{0} A \right) \left( \phi^{\otimes n} \right) \right] \, ,
\end{equation}
where $A \left( \phi \right) = \left( \left. \Gamma_k^{\left(2\right)} \right \vert_\phi + R_k \right)^{-1}$ and the $n = 0$ term is deleted because it does not contribute to observables.
We have also assumed that the sum may be taken out of the trace corresponding to an interchange of limits.
$\phi^{\otimes n}$ denotes the tensor product $\phi \otimes ... \otimes \phi$ with a total of $n$ factors.

Expanding the left hand side in powers of $\phi$ and comparing the coefficients leads to
\begin{equation}
\label{eq:WetterichExpandedIntermediate}
\partial_k \DifferentialNAt{n}{0} \Gamma_k \left( \phi^{\otimes n} \right)
=
\frac{1}{2} \Tr \left[ \left( \partial_k R_k \right) \left( \DifferentialNAt{n}{0} A \right) \left( \phi^{\otimes n} \right) \right] \, .
\end{equation}

We now wish to find an explicit expression for $\DifferentialNAt{n}{0} A$ which may be achieved inductively by noting that
\begin{equation}
\left( \DifferentialAt{\phi} A \right) \left( \psi \right) = - A \left( \phi \right) \circ \left( \DifferentialAt{\phi} \Gamma_k^{(2)} \right) \left( \psi \right) \circ A \left( \phi \right) \, ,
\end{equation}
or for short
\begin{equation}
\label{eq:PropagatorFirstDerivative}
\DifferentialAt{\phi} A = - A \circ \Differential \Gamma_k^{(2)} \circ A \, .
\end{equation}
An educated guess produces the induction hypothesis
\begin{equation}
\label{eq:PropagatorNthDerivative}
\Differential^n A =
\sum_{c \in \CombinationSet{n}} \left( -1 \right)^{\#c} \frac{n!}{c!}
A \circ
\prod_{l = 1}^{\# c}
\left[ \Differential^{c_l} \Gamma_k^{(2)} \circ A \right] \, ,
\end{equation}
where $\CombinationSet{n}$ denotes the set of all multi-indices with positive entries that are combinations\footnote{Partitions including permutations} of the natural number $n$, e.g
\begin{equation}
\CombinationSet{3} = \left \{
\left( 1, 1, 1 \right),
\left( 1, 2 \right),
\left( 2, 1 \right),
\left( 3 \right)
\right \} \, .
\end{equation}
In equation \ref{eq:PropagatorNthDerivative}, $\# c$ is the length of such a multi-index and
\begin{equation}
c! = \prod_{l = 1}^{\# c} \left( c_l! \right) \, ,
\qquad
\left \vert c \right \vert = \sum_{l = 1}^{\# c} c_l = n
\end{equation}
for all $n \in \mathbb{N}$ and any $c \in \CombinationSet{n}$.
The inductive proof of equation \ref{eq:PropagatorNthDerivative} is given in appendix \ref{apx:ProofPropagatorDerivative}.
Inserting this result into equation \ref{eq:WetterichExpandedIntermediate} then yields
\begin{equation}
\label{eq:WetterichExpanded}
\begin{aligned}
\partial_k &\DifferentialNAt{n}{0} \Gamma_k \left( \phi^{\otimes n} \right)
= \frac{1}{2} \sum_{c \in \CombinationSet{n}}
\left( -1 \right)^{\# c}
\frac{n!}{c!}
\Tr \Bigg \{ \\
&\left( \partial_k R_k \right)
A \left( 0 \right)
\prod_{l = 1}^{\# c}
\left[
\left( \DifferentialNAt{c_l}{0} \Gamma_k^{\left(2\right)} \right) \left( \phi^{\otimes c_l} \right)
A \left( 0 \right)
\right]
\Bigg \} \, .
\end{aligned}
\end{equation}
Equation \ref{eq:WetterichExpanded} expresses all possible one-loop diagrams generated by an arbitrary action $\Gamma_k$ contributing to the renormalisation group flow of a given correlation function.
As is common practice, we shall work with them explicitly in the Fourier picture\footnote{We define the Fourier transform $\tilde{f}$ of a measurable function $f : \mathbb{R}^d \to \mathbb{R}$ as \[ \tilde{f} \left( p \right) = \left( 2 \pi \right)^{-\frac{d}{2}} \int_{\mathbb{R}^d} \exp \left[ - i p x \right] f \left( x \right) \mathrm{d} x \] whenever the integral converges.}.
Restricting ourselves to translation-invariant quantum field theories, for every $n \in \mathbb{N}$ there is a ($k$-dependent) function $\kappa_n$\footnote{The prefactors $\left( 2 \pi \right)^{\frac{d}{2} \left( 2 - n \right)}$ are chosen such that they vanish in position space.} such that
\begin{equation}
\label{eq:KappaNDefinition}
\begin{aligned}
\left( \DifferentialNAt{n}{0} \Gamma_k \right) &\left( \phi_1 \otimes ... \otimes \phi_n \right)
= 
\left( 2 \pi \right)^{\frac{d}{2} \left( 2 - n \right)}
\int_{\left( \mathbb{R}^d \right)^{n-1}} \\
&\kappa_n \left( p_1, ..., p_{n-1}; k \right)
\tilde{\phi}_1 \left( p_1 \right) ... \tilde{\phi}_{n-1} \left( p_{n-1} \right) \\
&\tilde{\phi}_n \left( - \left[ p_1 + ... + p_{n - 1} \right] \right) \mathrm{d} p_1 ... \mathrm{d} p_{n-1}
\end{aligned}
\end{equation}
for all test functions $\phi_1, .., \phi_n$.
These $\kappa_n$ are precisely the commonly considered one-particle irreducible $n$-point functions in Fourier space stripped of their delta-functions:
\begin{equation}
\Gamma_k^{\left( n \right)} \left( p_1, ..., p_n \right) = \kappa_n \left( p_1, ..., p_{n-1} \right) \delta \left( p_1 + ... + p_n \right)
\end{equation}
Consequently, any such $\DifferentialNAt{n}{0} \Gamma_k$ is translation invariant in the sense that
\begin{equation}
\begin{aligned}
\left( \DifferentialNAt{n}{0} \Gamma_k \right) &\left( T \phi_1 \otimes ... \otimes T \phi_n \right) \\
&= \left( \DifferentialNAt{n}{0} \Gamma_k \right) \left( \phi_1 \otimes ... \otimes \phi_n \right)
\end{aligned}
\end{equation}
for all translations $T$ of $\mathbb{R}^d$ by the properties of the Fourier transform.
Furthermore, such a $\DifferentialNAt{n}{0} \Gamma_k$ is obviously $\mathcal{O} \left( d \right)$-invariant whenever the corresponding $\kappa_n$ is\footnote{The action of $\mathcal{O} \left( d \right)$ on a function $g : \left( \mathbb{R}^d \right)^n \to \mathbb{R}$ is the standard one, defined as \[ \left(O g\right) \left( p_1, ..., p_n \right) = g \left( O^{-1} p_1, ..., O^{-1} p_n \right) \] for all $O \in \mathcal{O} \left( d \right)$ and $p_1, ..., p_n \in \mathbb{R}^d$.}.
To simplify equations from this point on, any $k$-dependence will be suppressed whenever it does not lead to ambiguities.
Since Fréchet derivatives are invariant under permutations there are corresponding symmetries of the $\kappa_n$:
For all $\sigma \in \SymmetricGroupN{n-1}$
\begin{equation}
\kappa_n \left( p_{\sigma \left(1\right)}, ..., p_{\sigma \left(n-1\right)}\right) = \kappa_n \left( p_1, ..., p_{n-1} \right)
\end{equation}
and also
\begin{equation}
\label{eq:MomentumConservation}
\begin{aligned}
\kappa_n &\left( p_1, ..., p_{n-1} \right)
= \\
&\kappa_n \left( -\left[ p_1 + ... + p_{n-1} \right], p_2, ..., p_{n-1} \right)
\end{aligned}
\end{equation}
for all $p_1, ..., p_{n-1} \in \mathbb{R}^d$.
We shall refer to functions $f$ satisfying these symmetries as $\SymmetricGroupN{n-1}^*$ symmetric\footnote{\enquote{$\SymmetricGroup$} standing for the symmetric (permutation) group and \enquote{$\ast$} for the involution given by \[ \left( p_1, ..., p_{n-1} \right) \mapsto \left( -\left[ p_1 + ... + p_{n-1} \right], p_2, ..., p_{n-1} \right) \, . \] The full group $\SymmetricGroupN{n-1}^*$ of symmetries is isomorphic to $\SymmetricGroupN{n}$ but the underlying action is a non-standard one on $\left(n-1\right)$-tuples, hence the alternative naming.}.
It remains to phrase equation \ref{eq:WetterichExpanded} in terms of the correlation functions $\kappa_n$.
While the left-hand side is simple, let us look at the right-hand side first:
If the expression within the trace is viewed as an integral operator the trace can be evaluated by integration along the diagonal.
From the definition of the $\kappa_n$ we already know the integral form of the derivatives of $\Gamma_k$ and it only remains to express $R_k$ appropriately.
It is common practice to define $R_k$ in momentum space as a family of multiplication operators parametrised by $k$, i.e
\begin{equation}
\left[\mathcal{F} R_k \mathcal{F}^{-1} \phi \right] \left( p \right) = \bar{r} \left( p; k \right) \phi \left( p \right)
\end{equation}
for some $\bar{r} : \mathbb{R^d} \times \mathbb{R}_{> 0} \to \mathbb{R}$\footnote{We use $\bar{r}$ to avoid confusion with the commonly used shape function $r$ defined as \[ \bar{r} \left( p \right) = p^2 r \left( \frac{p^2}{k^2} \right) \, .\]}.
The role of $\bar{r}$ is to contribute a \enquote{momentum-dependent mass} that protects against IR singularities and at the same time screens UV divergences at any finite scale $k > 0$ by a rapid decay for large momenta.
Thus, $\bar{r}$ and $\kappa_2$ have to be treated on similar footings so that we have to demand
\begin{equation}
\label{eq:RegulatorIsSym1AstSymmetric}
\bar{r} \left( q \right) = \bar{r} \left( -q \right)
\end{equation}
for all $q \in \mathbb{R}^d$ in accordance with the $\SymmetricGroupN{1}^*$ symmetry of $\kappa_2$.
Choosing a $\SymmetricGroupN{1}^*$-violating regulator would generate further symmetry-breaking terms leading to undesirable contributions that are not translation invariant.
The trace in equation \ref{eq:WetterichExpanded} then becomes
\begin{equation}
\label{eq:WetterichExpandedRHSTrace}
\begin{aligned}
\Tr \left \{ ... \right \}
&=
\left( 2 \pi \right)^{- \frac{\left \vert c \right \vert d}{2}} 
\int_{\left(\mathbb{R}^d\right)^{\left \vert c \right \vert - 1}} 
\lambda_c \left( p_1, ..., p_{\left \vert c \right \vert - 1} \right) \\
&\phantom{=} \times
\tilde{\phi} \left( p_1 \right) ... \tilde{\phi} \left( p_{\left \vert c \right \vert - 1} \right) \\
&\phantom{=} \times
\tilde{\phi} \left( - \left[ p_1 + ... + p_{\left \vert c \right \vert - 1} \right] \right)
\mathrm{d} p_1 ... \mathrm{d} p_{\left \vert c \right \vert - 1} \, ,
\end{aligned}
\end{equation}
where
\begin{equation}
\begin{aligned}
\lambda_c &\left( p^1_1, ..., p^1_{c_1}, ..., p^{\# c}_{c_{\# c} - 1} \right) = 
\int_{\mathbb{R}^d} \frac{\left( \partial_k \bar{r} \right) \left( q \right)}{\left[ \kappa_2 \left( q \right) + \bar{r} \left( q \right) \right]^2} \\
&\kappa_{2 + c_{\# c}} \left(
p^{\# c}_1, ..., p^{\# c}_{c_{\# c} - 1}, - \sum_{a=1}^{\# c} \sum_{b=1}^{c_{\# c} - 1} p^a_b, q
\right) \\
&\prod_{l = 1}^{\# c -1}
\frac{\kappa_{2 + c_l} \left( p^l_1, ..., p^l_{c_l}, q - \sum_{a=1}^{l} \sum_{b=1}^{c_l} p^a_b \right)}{\left( \kappa_2  + \bar{r} \right) \left( q - \sum_{a=1}^{l} \sum_{b=1}^{c_l} p^a_b \right)} \mathrm{d} q \, .
\end{aligned}
\end{equation}
This represents an integral over an arbitrary one-loop diagram containing all possible vertices in closed form.

Let us now collect the above and rewrite equation \ref{eq:WetterichExpanded} as
\begin{equation}
\label{eq:WetterichExpandedIntermediate2}
\begin{aligned}
&0 = 
\int_{\left( \mathbb{R}^d \right)^{n-1}}
\Big[ \left( 2 \pi \right)^d \left( \partial_k \kappa_n \right) \left( p_1, ..., p_{n-1} \right) \\
&\phantom{\int_{\left( \mathbb{R}^d \right)^{n-1}} \Big[}
- \frac{1}{2} \sum_{c \in \CombinationSet{n}}
\left( -1 \right)^{\# c}
\frac{n!}{\; c!} \lambda_c \left( p_1, ..., p_{n-1} \right)
\Big] \\
&\phantom{0 = } \times
\tilde{\phi} \left( - \left[ p_1 + ... + p_{n - 1} \right] \right)
\tilde{\phi} \left( p_1 \right) ... \tilde{\phi} \left( p_{n-1} \right) \\
&\phantom{0 = }
\mathrm{d} p_1 ... \mathrm{d} p_{n-1} \, .
\end{aligned}
\end{equation}
Before the fundamental lemma of the calculus of variations may be invoked here, we need to polarise this equation, allowing for arbitrary test functions of the form $\phi_1 \otimes ... \otimes \phi_n$ instead of purely diagonal ones $\phi^{\otimes n}$.
However, this polarisation will leave the $\kappa_n$ part invariant (after proper substitutions of the integral variables) by equation \ref{eq:KappaNDefinition} due to its $\SymmetricGroupN{n-1}^\ast$ symmetry.
Therefore, such a polarisation is exactly the same as a $\SymmetricGroupN{n-1}^\ast$ symmetrisation.
Hence, simply defining
\begin{equation}
\bar{\lambda}_c \left( p_1, ..., p_{n-1} \right) = 
\frac{1}{n!} \sum_{\sigma \in \SymmetricGroupN{n}}
\lambda_c \left( p_{\sigma \left( 1 \right)}, ..., p_{\sigma \left( n- 1 \right)} \right)
\end{equation}
where we set $p_n = - \left[ p_1 + ... + p_{n-1} \right]$ sidesteps the explicit polarisation.
Invoking the fundamental lemma of the calculus of variations then leads to
\begin{equation}
\label{eq:WetterichExpandedWithLambdaBars}
\partial_k \kappa_n = \frac{1}{2 \left( 2 \pi \right)^d} \sum_{c \in \CombinationSet{n}} \left( -1 \right)^{\# c}  \frac{n!}{c!} \bar{\lambda}_c \, .
\end{equation}
This is an equivalent formulation of equation \ref{eq:WetterichExpanded} and will be referred to as the flow equation of the correlation function $\kappa_n$.
While these equations for arbitrary $n \in \mathbb{N}$ are certainly implied by the $\Lambda$-free form of equation \ref{eq:WetterichEquation} if
\begin{itemize}
\item $\Gamma_k$ is analytic,
\item $\phi \mapsto \left( \left . \Gamma_k^{\left(2\right)} \right \vert_{\phi} + R_k \right)^{-1}$ is analytic,
\item the sum in equation \ref{eq:WetterichExpandedFull} may be pulled out of the trace
\end{itemize}
the converse is not necessarily true:
A given solution might not correspond to an analytic $\Gamma_k$, that is the formal series
\begin{equation}
\begin{aligned}
\sum_{n = 1}^\infty
&\frac{\left( 2 \pi \right)^{\frac{d}{2} \left( 2 - n \right)}}{n!}
\int_{\left( \mathbb{R}^d \right)^{n-1}}
\kappa_n \left( p_1, ..., p_{n-1} \right) \\
&\times \tilde{\phi} \left( - \left[ p_1 + ... + p_{n-1} \right] \right) \tilde{\phi} \left( p_1 \right) ... \tilde{\phi} \left( p_{n-1} \right) \\
&\mathrm{d} p_1 ... \mathrm{d} p_{n-1}
\end{aligned}
\end{equation}
might diverge for some non-zero test function $\phi$.
Nonetheless, in the study of differential equations a lot of insight is often gained by an initial broadening of the space of admissible solutions and in this spirit, one might even expect such formal solutions to be very important for the general study of the Wetterich equation.

One further remark is in order at this point: Upon solving equation \ref{eq:WetterichEquation} it is not clear whether there always exists a corresponding $S_\Lambda$ satisfying equation \ref{eq:PathIntegralRepresentation} amounting to the reconstruction problem \cite{src:ReuterSaueressig2019}.
Especially, a possible non-uniqueness of solutions to equation \ref{eq:WetterichEquation} casts doubts on a positive conjecture.
The situation is made even less clear by studying solutions to the $\Lambda$-free version of the Wetterich equation due to the difficulty of non-regularised path integrals.
\section{A Constructive Solution for the Correlation Functions}
\label{sec:AConstructiveSolution}
A full solution to the flow equations \ref{eq:WetterichExpandedWithLambdaBars} with $\kappa_m \neq 0$ for some $m \in \mathbb{N}_{\ge 3}$ of course seems rather difficult to find due to the non-linear structure of the $\bar{\lambda}_c$ terms.
This is the reason why one in practise usually truncates the equations at a finite $n \in \mathbb{N}$.
There are, however, precisely three terms on the right hand side of equation \ref{eq:WetterichExpandedWithLambdaBars} for $n \in \mathbb{N}_{\ge 3}$ revealing a somewhat linearish structure, namely
\begin{equation}
\begin{aligned}
c &= \left( n \right) &\Rightarrow \bar{\lambda}_c \text{ depends linearly on } \kappa_{n + 2} \\
c &= \left( n - 1, 1 \right) &\Rightarrow \bar{\lambda}_c \text{ depends linearly on } \kappa_{n + 1} \\
c& = \left( 1, n - 1 \right) &\Rightarrow \bar{\lambda}_c \text{ depends linearly on } \kappa_{n + 1}
\end{aligned}
\end{equation}
Phrased differently, for all $n \in \mathbb{N}_{\ge 3}$ there exist linear operators $I_n$ implicitly depending on $\left \{ \kappa_2, \bar{r} \right \}$ and $J_n$ implicitly depending on $\left \{ \kappa_2, \kappa_3, \bar{r} \right \}$ such that
\begin{equation}
\label{eq:CorrelatorFlowInKappanp2}
\begin{aligned}
I_n \kappa_{n+2} &= - 2 \left( 2 \pi \right)^d \partial_k \kappa_n + n J_n \kappa_{n+1}\\
& + \sum_{c \in \CombinationSet{n} \setminus \left \{ \left( n \right), \left( n - 1, 1 \right), \left( 1, n - 1 \right) \right \}} \left( -1 \right)^{\# c}  \frac{n!}{c!} \bar{\lambda}_c \, .
\end{aligned}
\end{equation}
The significance of this equation lies in the fact, that the right-hand side depends only on $\left \{ \kappa_2, ..., \kappa_{n + 1}, \bar{r} \right \}$.
Suppose now, that all $I_n$ possess right inverses $\rho_n$, i.e mappings such that $I_n \circ \rho_n = \mathrm{id}$.
Then, setting
\begin{equation}
\label{eq:Kappanp2FromRightInverse}
\begin{aligned}
\kappa_{n+2} &= \rho_n \Bigg[ - 2 \left( 2 \pi \right)^d \partial_k \kappa_n + n J_n \kappa_{n+1} \\
&\phantom{=} + \sum_{c \in \CombinationSet{n} \setminus \left \{ \left( n \right), \left( n - 1, 1 \right), \left( 1, n - 1 \right) \right \}} \left( -1 \right)^{\# c}  \frac{n!}{c!} \bar{\lambda}_c \Bigg]
\end{aligned}
\end{equation} 
will evidently solve equation \ref{eq:CorrelatorFlowInKappanp2}.
This fact suggests the following approach for solving the flow equation for the correlators:
\begin{enumerate}
\item For some $N \in \mathbb{N}$ find $\kappa_1, ..., \kappa_{N + 1}$ satisfying equation \ref{eq:WetterichExpandedWithLambdaBars} for all $n \in \mathbb{N}_{< N}$
\item Find a right inverse $\rho_N$ of $I_N$
\item Construct $\kappa_{N+2}$ as in equation \ref{eq:Kappanp2FromRightInverse}
\item Increase $N$ by $1$ and go back to step $2$
\end{enumerate}
This iterative construction will produce $\kappa_n$ for all $n \in \mathbb{N}$ and they will satisfy their respective flow equation.
Evidently, this construction depends crucially on the initial $\kappa_1, ..., \kappa_{N + 1}$ which have to be given as input for all values of momenta and the scale $k$.
This input may be pictured as a different kind of boundary condition to the Wetterich equation:
Instead of specifying the classical theory at $k \to \infty$ one specifies the full renormalisation group flow of a finite set of one-particle irreducible correlators which is exemplified in figure \ref{fig:idea}.
\begin{figure}
\includegraphics[scale=0.22]{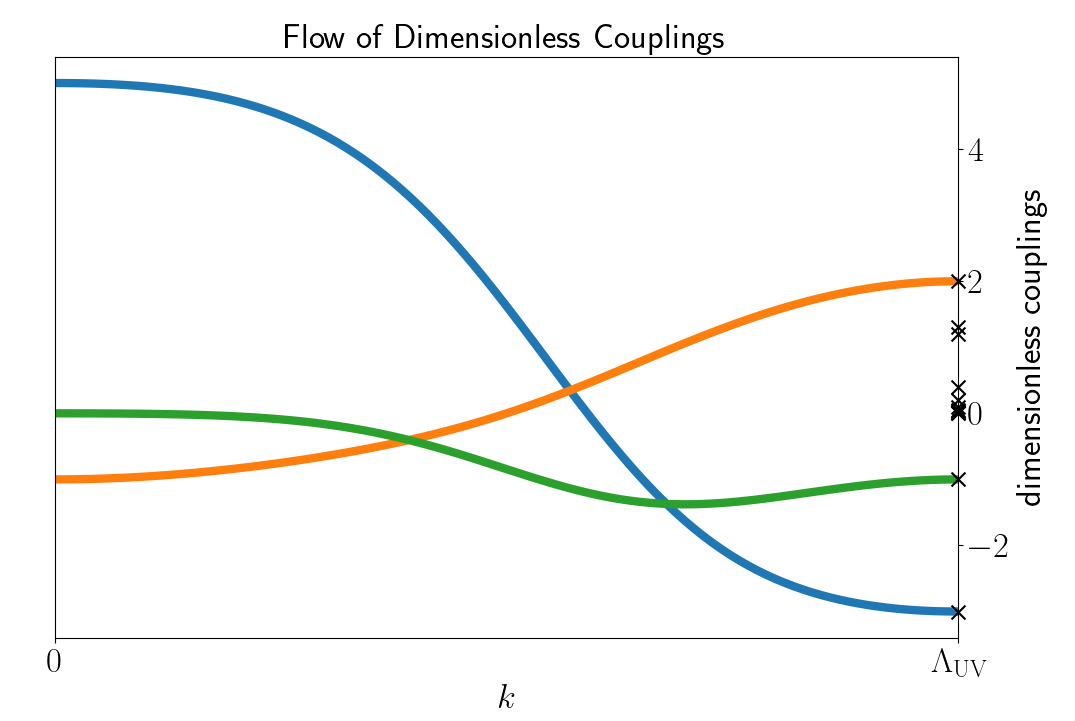}
\caption{A fictitious renormalisation group flow from a possibly infinite UV scale $\Lambda_{\mathrm{UV}}$ to $0$. In the presented approach initial conditions correspond to the exemplified flow of the three dimensionless couplings. In the traditional approach the initial conditions are given by the values of infinitely many couplings at $\Lambda_{\mathrm{UV}}$ exemplified by the crosses on the right axis.}
\label{fig:idea}
\end{figure}
A prototypical input might be given by a propagator $\kappa_2$ obtained through through some other well-developed method like a derivative expansion or even a numerical lattice computation.
In case of a $\mathbb{Z}_2$-symmetric theory such input in fact provides a full starting point for the presented scheme since $\kappa_1$ and $\kappa_3$ vanish.

As is clear from this discussion, there is a certain amount of choice involved.
Furthermore, in every iteration there may be several right inverses to choose from because the kernel of any $I_n$ might be non-empty.
Hence, this procedure is quite different from the usual approach of giving specific boundary conditions at some scale $k$ or at $k \to \infty$.
In fact, it shall be demonstrated that imposing the naive boundary condition of a real scalar $\phi^4$ theory in the UV limit $k \to \infty$ does not guarantee the uniqueness of solutions to equation \ref{eq:WetterichExpandedWithLambdaBars}.
However, before diving into the specifics of $\phi^4$ theory, we shall give explicit expressions for the $I_n$ and particularly simple choices of linear right inverses $\rho_n$.
For brevity, define
\begin{equation}
K \left( q \right) =
\frac{\left( \partial_k \bar{r} \right) \left( q \right)}{\left[ \kappa_2 \left( q \right) + \bar{r} \left( q \right) \right]^2} \, ,
\end{equation}
allowing to write
\begin{equation}
\begin{aligned}
&\lambda_{\left( n \right)} \left( p_1, ..., p_{n-1} \right)
=
\int_{\mathbb{R}^d} K \left( q \right) \\
&\times \kappa_{n+2} \left( p_1, ..., p_{n-1}, - \left[ p_1 + ... + p_{n-1} \right], q \right)
\mathrm{d} q \, .
\end{aligned}
\end{equation}
By the $\SymmetricGroupN{n+1}^*$ symmetry of $\kappa_{n+2}$, this is $\SymmetricGroupN{n-1}^*$ symmetric such that $\bar{\lambda}_{\left( n \right)} = \lambda_{\left( n \right)}$ and
\begin{equation}
\begin{aligned}
\bar{\lambda}_{\left( n \right)} &\left( p_1, ..., p_{n-1} \right) \\
&=
\int_{\mathbb{R}^d} K \left( q \right)
\kappa_{n+2} \left( p_1, ..., p_{n-1}, - q, q \right)
\mathrm{d} q \, .
\end{aligned}
\end{equation}
Thus, one may write
\begin{equation}
\begin{aligned}
\left( I_n f \right) ( p_1, &..., p_{n-1} )
= \\
&\int_{\mathbb{R}^d}
K \left( q \right) f \left( p_1, ..., p_{n-1}, -q, q \right)
\mathrm{d} q
\end{aligned}
\end{equation}
for all functions $f : \left( \mathbb{R}^d \right)^{n+1} \to \mathbb{R}$ where the integral exists. The reason for allowing arbitrary functions $f$ and not just $\SymmetricGroupN{n+1}^*$-symmetric ones is to facilitate the proof given in appendix \ref{apx:ProofRhonIsRightInverse} that the yet to be defined $\rho_n$ are indeed right inverses of the corresponding $I_n$.
An obvious choice of linear right inverse of the above $I_n$ is given by
\begin{equation}
\left( \bar{\rho}_n g \right) \left( p_1, ..., p_{n+1} \right)
=
\frac{g \left( p_1, ..., p_{n-1} \right)}{\int_{\mathbb{R}^d} K}
\end{equation}
However, in general such $\bar{\rho}_n g$ will not be $\SymmetricGroupN{n+1}^*$ symmetric whenever $g$ is $\SymmetricGroupN{n-1}^*$ symmetric which is unacceptable here, as it would generate terms that are not momentum conserving.
Taking $\bar{\rho}_n$ as an ansatz and successively eliminating all $\SymmetricGroupN{n+1}^*$-violating terms generated by the action of $\SymmetricGroupN{n+1}^*$ on functions of the form $\bar{\rho}_n g$ where $g$ is taken to be $\SymmetricGroupN{n-1}^*$ symmetric leads to the much better choice
\begin{equation}
\label{eq:rhonDefinition}
\begin{aligned}
\left( \rho_n g \right) &\left( p_1, ..., p_{n+1} \right) =
\sum_{J \subseteq \left \{ 0, ..., n + 1 \right \}}
\sum_{l = 0}^{\left \lfloor \frac{n - 1 - \# J}{2} \right \rfloor} \\
&\frac{\alpha^{n}_{\# J, l}}{\left( \int_{\mathbb{R}^d} K \right)^{n - \# J - l}} 
\int_{\left( \mathbb{R}^d \right)^{n - 1 - \# J - l}} \\
&g \left( p_J, -s_1, s_1, ..., -s_l, s_l, t_1, ..., t_{n - 1 - \# J -2l} \right) \\
&\times K \left( s_1 \right) ... K \left( s_l \right)
K \left( t_1 \right) ... K \left( t_{n - 1 - \# J - 2l} \right) \\
&\mathrm{d} s_{...} \mathrm{d} t_{...}
\end{aligned}
\end{equation}
with
\begin{equation}
\label{eq:alphanabDefinition}
\alpha^{n}_{a,b} = \frac{\left( - 1 \right)^{n - 1 - a - b}}{n} 2^{n - 1 - a - 2b} \binom{n - 1 - a - b}{b} \, .
\end{equation}
In the above expression we have defined $p_0 = - \left[ p_1 + ... + p_{n+1} \right]$ and introduced the shorthand notation $p_J := p_{J_1}, ..., p_{J_{\# J}}$.
Note that the particular order of the corresponding momenta $p$ in the above expression does not matter since $g$ is presumed symmetric.
Hence, we do not need another sum over all permutations of index sets $J$.
For a proof that $\rho_n$ when restricted to $\SymmetricGroupN{n-1}^*$-symmetric functions is indeed a right inverse of $I_n$ see appendix \ref{apx:ProofRhonIsRightInverse}.

It is obvious that $\rho_n$ is a linear operator and thus a particular simple choice of right inverse of $\kappa_n$.
Furthermore, it preserves $\mathcal{O} \left( d \right)$-invariance provided $K$ itself is $\mathcal{O} \left( d \right)$-invariant.
In our naive approach to $\phi^4$ theory, we shall consider a two-point function that does not scale with $k$ and approximates the free propagator
\begin{equation}
\kappa_{2,\text{free}} \left( p \right) = m^2 + \left \Vert p \right \Vert^2
\end{equation}
for some mass $m$.
Hence, any $k$ scaling of $K$ comes from the choice of a regulator.
Furthermore, common regulators scale like $k^2$ at small momenta leading to an overall $k$ scaling of $K$ as $k^{-3}$.
A simple power counting in equation \ref{eq:rhonDefinition} then reveals that $\rho_n$ scales like $k^{3-d}$.
This fact is remarkable, as it indicates that in $d > 3$ dimensions the correlators constructed through $\rho_n$ are strongly suppressed for large $k$.
This simplifies the control of the \enquote{classical limit} $k \to \infty$, as one usually considers only a finite set of non-zero correlation functions in this limit.
The small $k$ behaviour is precisely the opposite.
Here $\rho_n$ grows arbitrarily large, possibly leading to IR divergences.

As mentioned before, the choice of a right inverse is not necessarily unique which can be seen explicitly in the case of $n = 2$.
With the previous construction, we have
\begin{equation}
\begin{aligned}
&\left( \rho_2 g \right) \left( p, q, r \right)
\\
&= \frac{1}{2 \int_{\mathbb{R}^d} K} \left[ g\left( p \right) + g\left( q \right) + g\left( r \right) + g\left( - p - q - r \right) \right] \\
&\phantom{=} - \frac{1}{\left( \int_{\mathbb{R}^d} K \right)^2} \int_{\mathbb{R}^d} g \left( t \right) K \left( t \right) \mathrm{d} t \, ,
\end{aligned}
\end{equation}
which satisfies $\left( I_2 \circ \rho_2 \right) g = g$ whenever $g$ is $\SymmetricGroupN{1}^*$ symmetric.
However, there also exists a suitable non-linear right inverse $\rho_2'$ given by
\begin{equation}
\begin{aligned}
&\left( \rho_2' g \right) \left( p, q, r \right) = \frac{1}{8 \int_{\mathbb{R}^d} g K} \Big( \\
&\times \Big(
\left[ g\left( p \right) + g\left( q \right) + g\left( r \right) + g\left( - p - q - r \right) \right]^2 \\
&\;
- 2 \left[ g\left( p \right)^2 + g\left( q \right)^2 + g\left( r \right)^2 + g\left( - p - q - r \right)^2 \right]
\Big) \, .
\end{aligned}
\end{equation}
Hence, the operator $I_2$ indeed has a non-trivial kernel, since $I_2 \circ \left( \rho_2 - \rho_2' \right) = 0$.
Thus, there is a certain degree of freedom involved in the choice of a right inverse to $I_2$.
In particular this choice may be used to construct higher correlators that satisfy certain constraints such as boundary conditions (e.g at $k \to 0$ or $k \to \infty$) or decay properties like those produced in \cite{src:Pawlowski2007}.

\section{Solving the Flow Equations}
\label{sec:SolvingTheFlowEquations}
We shall consider a real scalar quantum field theory in $d$ Euclidean dimensions without spontaneous symmetry breaking with the \enquote{classical limit}
\begin{equation}
\label{eq:initialConditions}
\begin{aligned}
\lim_{k \to \infty} \kappa_2 \left( p \right) = \kappa_{2,\text{free}} \left( p \right) &= m^2 + \left \Vert p \right \Vert^2 \\
\lim_{k \to \infty} \kappa_4 \left( p, q, r \right) &= \frac{\lambda}{\left \vert m \right \vert^{d-4}} \\
\forall n \in \mathbb{N} \setminus \left \{ 2, 4 \right \} : \lim_{k \to \infty} \kappa_n &= 0
\end{aligned}
\end{equation}
for some $m \in \mathbb{R}$, $\lambda > 0$\footnote{The $\kappa_4$ limit has been chosen such that $\lambda$ is dimensionless.} where the limits should be understood in a distributional sense\footnote{Technically speaking, $\kappa_n$ is a distribution on $\mathbb{R}^{\left(n-1\right)d}$ and the $k$ limits should be understood as pointwise convergence.}.
In particular, for $k \to \infty$ all odd correlation functions vanish.
We shall now set $N = 3$ and proceed as outlined in the preceding section.
The reason for setting $N = 3$ is of course to be able to satisfy the boundary condition for $\kappa_{N+1} = \kappa_4$ for $k \to \infty$.
We thus choose the ansatz
\begin{equation}
\label{eq:Kappa4Definition}
\begin{aligned}
\kappa_4 &\left( p, q, r; k \right) = \frac{\lambda}{\left \vert m \right \vert^{d-4}} \exp \Bigg[ \\
& - \frac{\left \Vert p \right \Vert^d + \left \Vert q \right \Vert^d + \left \Vert r \right \Vert^d + \left \Vert p + q + r \right \Vert^d + \left \vert m \right \vert^d}{k \left \vert m \right \vert^{d-1}} \Bigg]
\end{aligned}
\end{equation}
which is obviously $\SymmetricGroupN{3}^\ast$ and $\mathcal{O} \left( d \right)$ invariant and satisfies equation \ref{eq:initialConditions}.
The rationale for choosing this particular form for $\kappa_4$ is to keep the upcoming integrals as simple as possible and to ensure a rapid decrease of $\kappa_4$ and its $k$ derivatives for $k \to 0$.
The latter is paramount for controlling the divergent $k$ behaviour of $\rho_n$ in this limit.
At the same time, all higher correlators as generated by the $\rho_n$ will vanish in the UV due to the very same $k$-scaling.
The most natural choice for the lower odd correlators is
\begin{equation}
\kappa_3 = 0 \qquad \text{and} \qquad \kappa_1 = 0 \, ,
\end{equation}
which alongside the given construction procedure guarantees the vanishing of all odd correlators because
\begin{itemize}
\item for all odd $n \in \mathbb{N}$ any $c \in \CombinationSet{n}$ contains an odd entry,
\item the chosen $\rho_n$ are linear.
\end{itemize}
This implements the standard $\mathbb{Z}_2$ symmetry such that only even correlators have to be dealt with.
Equation \ref{eq:Kappanp2FromRightInverse} then simplifies to
\begin{equation}
\label{eq:Kappa2np2FromRightInverse}
\begin{aligned}
\kappa_{2n+2} &= \rho_{2n} \Bigg[ - 2 \left( 2 \pi \right)^d \partial_k \kappa_{2n} \\
&\phantom{=} + \sum_{c \in \EvenCombinationSet{2n} \setminus \left \{ \left( 2n \right) \right \}} \left( -1 \right)^{\# c}  \frac{\left( 2 n \right)!}{c!} \bar{\lambda}_c \Bigg] \, ,
\end{aligned}
\end{equation} 
where $\EvenCombinationSet{n} \subset \CombinationSet{n}$ denotes the set of combinations with even entries.
The next step is now to find $\kappa_2$, since the flow equation for $\kappa_1$ is trivially satisfied.
Equation \ref{eq:WetterichExpandedWithLambdaBars} for $n = 2$ reads
\begin{equation}
\begin{aligned}
\partial_k \kappa_2 \left( p \right) &=
- \frac{1}{2 \left( 2 \pi \right)^d} \left( I_2 \kappa_4 \right) \left( p \right) \\
&=
- \frac{1}{2 \left( 2 \pi \right)^d} \int_{\mathbb{R}^d} K \left( q \right)
\kappa_{4} \left(p, - q, q \right)
\mathrm{d} q \, ,
\end{aligned}
\end{equation}
which in general cannot be expected to have a solution that can be put in closed form due to the dependence of $K$ on $\bar{r}$ and $\kappa_2$.
One may, however, show that the differential equation may be solved iteratively as is done in appendix \ref{apx:ProofKappa2IterationConverges}.
The initial ansatz is chosen to be the free propagator $\kappa_{2,\text{free}}$ and the regulator is chosen as
\begin{equation}
\label{eq:RegulatorDefinition}
\bar{r} \left( p; k \right) = \frac{\left \Vert p \right \Vert^2}{\exp \left[ \frac{\left \Vert p \right \Vert^2}{k^2} \right] - 1} \, ,
\end{equation}
both of which are $\mathcal{O} \left( d \right)$-invariant.
It is then demonstrated that whenever
\begin{equation}
\label{eq:lambdaBoundFromIteration}
0 \le \lambda < \left( \sqrt{3} - 1 \right) d \; 2^{d+1} \pi^{\frac{d}{2} - 1} \Gamma\left( \frac{d}{2} \right) \, ,
\end{equation}
a bounded, $\mathcal{O} \left( d \right)$-invariant and smooth (in its momentum argument as well as in $k$) solution satisfying the boundary condition \ref{eq:initialConditions} exists and is approached by the iterative scheme.
Note that the upper bound for $\lambda$ does not denote a critical coupling, it merely ensures that rather straightforward estimates may be applied.
We shall henceforth assume $\lambda$ to be bounded as in equation \ref{eq:lambdaBoundFromIteration}.
At the core of the proof lies the inequality
\begin{equation}
\label{eq:inversePropagatorEstimate}
\frac{1}{m^2 + \left \Vert p \right \Vert^2 + \bar{r} \left( p \right)}
\le
\frac{1}{m^2 + k^2} \, ,
\end{equation}
leading to the existence of a $\kappa_2 > m^2$ satisfying
\begin{equation}
\label{eq:Kappa2isAFixedPointOfIteration}
\begin{aligned}
\kappa_2 &\left( p; k \right) = \kappa_{2,\text{free}} \left( p; k \right) +
\frac{1}{2} \left( 2 \pi \right)^{-d}
\int_k^\infty
\int_{\mathbb{R}^d} \\
&\frac{\partial_{k'} r \left( q; k' \right)}{\left[ \kappa_2 \left( q; k' \right) + \bar{r} \left( q; k' \right) \right]^2}
\kappa_{4} \left(p, - q, q; k' \right)
\mathrm{d} q \,
\mathrm{d} k' \, .
\end{aligned}
\end{equation}
The iterative construction procedure of $\kappa_2$ also guarantees the existence of the IR limit $k \to 0$.
Note however, that in this limit $\kappa_2$ does not correspond to the free propagator.
Once we know $\kappa_2$, constructing the higher-order correlation functions is straightforward employing equation \ref{eq:Kappa2np2FromRightInverse}.
Their respective $\mathcal{O} \left( d \right)$-invariance follows from that of $K$.
It remains to discuss the behaviour of the correlators in the IR limit $k \to 0$ and the UV limit $k \to \infty$ respectively:
Obviously $\kappa_4$ vanishes in the limit of $k \to 0$.
As is proved in appendix \ref{apx:BigProof} for all $n \in \mathbb{N}_{\ge 2}$ there are constants $B_{2n}^{0,1} > 0$ such that
\begin{equation}
\label{eq:Kappa2nSupBound}
\left \Vert \kappa_{2n} \right \Vert_{L^\infty}
\le
B_{2n}^{0,1}
\frac{\left \vert m \right \vert^{2 + \left(2 - d \right) \left( n - 1 \right) + \left( n - 2 \right) \left( 1 + \Delta \right)} k}{\left( k + \left \vert m \right \vert \right)^{\left( n - 2 \right) \left( 1 + \Delta \right) + 1}}
\end{equation}
for
\begin{equation}
\label{eq:DeltaDefinition}
\Delta = \begin{cases}
1 & d \ge 4 \\
d - 3 & d < 4
\end{cases} \, .
\end{equation}
These equations establish the central result of this work:
For $d > 2$ all higher correlators vanish in both limits $k \to 0$ and $k \to \infty$.
Thus, the IR limit is a non-interacting theory with a non-trivially momentum dependent propagator $\kappa_2$ - a generalised free theory.
It may also be possible that the given solutions generalise to $d = 2$, since the proofs only make use of the property that the UV behaviour of $\left \vert \partial_k^l \kappa_4 \right \vert$ is bounded by $\sim k^{-l}$.
It is, however, even bounded by $\sim k^{-l-1}$ whenever $l \in \mathbb{N}$ which should guarantee the correct UV limits, while equation \ref{eq:Kappa2nSupBound} still ensures trivial IR limits.
A formal argument showing this has not yet been produced.

In the definition of $\kappa_4$ in equation \ref{eq:Kappa4Definition}, note that the argument in the exponential can be multiplied by any positive real number and still all estimates hold analogously with modified constants.
Furthermore, the boundary conditions at $k \to \infty$ remain satisfied and all higher correlators vanish at $k = 0$ upon such a modification of $\kappa_4$.
At the same time, the IR limit of $\kappa_2$ will in general be different.
Such ansatzes do not correspond to a rescaling of $k$ since the $k$ dependence of $\bar{r}$ remains unaltered.
Instead, they lead to different flows solving the flow equations for the correlators.
\section{The Flow of the Dimensionless Potential}
It is possible to extract the quantum potential from the correlators by examining their behaviour at zero momentum.
Of particular interest is the flow of the dimensionless potential $v$ given by
\begin{equation}
v \left( s \right):= \sum_{n = 1}^\infty \frac{\kappa_{2n} \left( 0, ..., 0 \right)}{k^{2 + \left(2-d\right)\left(n-1 \right)}} \frac{s^{2n}}{\left( 2 n \right)!} \, .
\end{equation}
It is appropriate to analyse its dimensionless flow, i.e $k \partial_k v$ which we shall examine in the limits $k \to 0$ and $k \to \infty$.
The $\kappa_2$ contribution is determined by equation \ref{eq:Kappa2isAFixedPointOfIteration} where the second term on the right-hand side is non-negative for all $p \in \mathbb{R}^d$.
Hence,
\begin{equation}
\lim_{k \to 0} \frac{\kappa_2 \left( 0 \right)}{k^2} \ge \lim_{k \to 0} \frac{\kappa_{2,\text{free}} \left( 0 \right)}{k^2} = \infty
\end{equation}
so that the resulting two-point correlator contains a gap that is bounded from below by the bare gap.
Furthermore,
\begin{equation}
\begin{aligned}
&\lim_{k \to 0} k \partial_k \frac{\kappa_2 \left( 0 \right)}{k^2}
\le
\lim_{k \to 0} \left[ \frac{\left \Vert \partial_k \kappa_2 \right \Vert_{L^\infty}}{k} - 2 \frac{\kappa_{2,\text{free}} \left( 0 \right)}{k^2} \right] \\
&\le
\lim_{k \to 0} \left[ \frac{\left( 2 \pi \right)^{-d}}{2} R_1 A_4^0 \frac{\left \vert m \right \vert^3 k}{\left( k^2 + m^2 \right)^2} - 2 \frac{\kappa_{2,\text{free}} \left( 0 \right)}{k^2} \right] \\
&= - \infty
\end{aligned}
\end{equation}
for constants $R_1, A_4^0 \ge 0$, where the $\left \Vert \partial_k \kappa_2 \right \Vert_{L^\infty}$ estimate is taken from equation \ref{eq:KDerivativeKappa2Estimate} in appendix \ref{apx:BigProof}.
Thus, the contribution of the propagator to the dimensionless potential diverges in the limit of $k \to 0$ which may be expected, since $m$ is taken to not scale with $k$.
The UV limits become
\begin{equation}
\lim_{k \to \infty} \frac{\kappa_2 \left( 0 \right)}{k^2} = \lim_{k \to \infty} \frac{m^2}{k^2} = 0
\end{equation}
and
\begin{equation}
\begin{aligned}
\lim_{k \to \infty} &k \partial_k \frac{\kappa_2 \left( 0 \right)}{k^2}
\le
\lim_{k \to \infty} \left[ \frac{\left \Vert \partial_k \kappa_2 \right \Vert_{L^\infty}}{k} - 2 \frac{m^2}{k^2} \right] \\
&\le
\lim_{k \to \infty} \left[ \frac{\left( 2 \pi \right)^{-d}}{2} R_1 A_4^0 \frac{\left \vert m \right \vert^3 k}{\left( k^2 + m^2 \right)^2} - 2 \frac{m^2}{k^2} \right] \\
&= 0 \, .
\end{aligned}
\end{equation}
Thus in the limit of $k \to \infty$ the corresponding contribution to $v$ vanishes and the solution lives in the deep-Euclidean region.
For the contributions from the higher correlators, we use theorem \ref{thm:Kappa2nDerivativesEstimate} from appendix \ref{apx:BigProof} to produce the estimates
\begin{align}
\left \Vert \kappa_{2n} \right \Vert_{L^\infty}
&\le
B_{2n}^{0,x}
\frac{\left \vert m \right \vert^{2 + \left(2 - d \right) \left( n - 1 \right) + \left( n - 2 \right) \left( 1 + \Delta \right)} k^x}{\left( k + \left \vert m \right \vert \right)^{\left( n - 2 \right) \left( 1 + \Delta \right) + x}} \, , \\
\left \Vert \partial_k \kappa_{2n} \right \Vert_{L^\infty}
&\le
B_{2n}^{1,x}
\frac{\left \vert m \right \vert^{2 + \left(2 - d \right) \left( n - 1 \right) + \left( n - 2 \right) \left( 1 + \Delta \right)} k^x}{\left( k + \left \vert m \right \vert \right)^{\left( n - 2 \right) \left( 1 + \Delta \right) + 1 + x}}
\end{align}
with constants $B_{2n}^{0,x}, B_{2n}^{1,x} \ge 0$ for all $x \in \mathbb{N}$ and $n \in \mathbb{N}_{\ge 2}$.
Hence, for all such $n$,
\begin{equation}
\begin{aligned}
\Big \vert k \partial_k &\frac{\kappa_{2n} \left( 0, ..., 0 \right)}{k^{2 + \left(2-d\right)\left( n - 1 \right)}} \Big \vert
\le \left \vert \frac{\partial_k \kappa_{2n} \left( 0, ..., 0 \right)}{k^{1 + \left(2-d\right)\left( n - 1 \right)}} \right \vert \\
& + \left \vert \frac{2 + \left(2-d\right)\left( n - 1 \right)}{k^{2 + \left(2-d\right)\left( n - 1 \right)}} \kappa_{2n} \left( 0, ..., 0 \right) \right \vert \, .
\end{aligned}
\end{equation}
With the previous inequalities, we then obtain
\begin{equation}
\begin{aligned}
&\lim_{k \to 0} \left \vert \frac{\kappa_{2n} \left( 0, ..., 0 \right)}{k^{2 + \left(2-d\right)\left( n - 1 \right)}} \right \vert \le 
B_{2n}^{0, \max \left \{ 1, 3 + \left( 2 - d \right) \left( n - 1 \right) \right \}} \\
&\phantom{\le} \times
\left( \frac{\left \vert m \right \vert}{k} \right)^{2 + \left( 2 - d \right) \left( n - 1 \right) - \max \left \{ 1, 3 + \left( 2 - d \right) \left( n - 1 \right) \right \}} 
= 0 \, .
\end{aligned}
\end{equation}
Likewise
\begin{equation}
\begin{aligned}
&\lim_{k \to 0} \left \vert \frac{\partial_k \kappa_{2n} \left( 0, ..., 0 \right)}{k^{1 + \left(2-d\right)\left( n - 1 \right)}} \right \vert \le 
B_{2n}^{1, \max \left \{ 1, 2 + \left( 2 - d \right) \left( n - 1 \right) \right \}} \\
&\phantom{\le} \times
\left( \frac{\left \vert m \right \vert}{k} \right)^{1 + \left( 2 - d \right) \left( n - 1 \right) - \max \left \{ 1, 2 + \left( 2 - d \right) \left( n - 1 \right) \right \}} 
= 0 \, ,
\end{aligned}
\end{equation}
so that $v$ and $k \partial_k v$ in the limit of small $k$ are fully determined by the $\kappa_2$ contributions.
For large $k$ the estimates
\begin{align}
\left \vert \frac{\kappa_{2n} \left( 0, ..., 0 \right)}{k^{2 + \left(2-d\right)\left( n - 1 \right)}} \right \vert &\le 
B_{2n}^{0, 1}
\left( \frac{\left \vert m \right \vert}{k} \right)^{\left( 3 - d + \Delta \right) \left ( n - 2 \right) + 4 - d} \\
\left \vert \frac{\partial_k \kappa_{2n} \left( 0, ..., 0 \right)}{k^{1 + \left(2-d\right)\left( n - 1 \right)}} \right \vert &\le 
B_{2n}^{1, 1}
\left( \frac{\left \vert m \right \vert}{k} \right)^{\left( 3 - d + \Delta \right) \left ( n - 2 \right) + 4 - d}
\end{align}
produce meaningful bounds whenever $d \le 4$:
\begin{align}
\lim_{k \to \infty} \left \vert \frac{\kappa_{2n} \left( 0, ..., 0 \right)}{k^{2 + \left(2-d\right)\left( n - 1 \right)}} \right \vert &\le \begin{cases}
0 & d \le 4 \\
B_{2n}^{0,1} & d = 4 \\
\infty & \text{otherwise} \, ,
\end{cases} \\
\left \vert \frac{\partial_k \kappa_{2n} \left( 0, ..., 0 \right)}{k^{1 + \left(2-d\right)\left( n - 1 \right)}} \right \vert &\le \begin{cases}
0 & d \le 4 \\
B_{2n}^{1,1} & d = 4 \\
\infty & \text{otherwise} \, .
\end{cases}
\end{align}
Thus,
\begin{equation}
\lim_{k \to \infty} \left \vert v \left( s \right) \right \vert \le \begin{cases}
0 & d < 4 \\
\sum_{n = 2}^\infty B_{2n}^{0,1} \frac{s^{2n}}{\left( 2n \right)!} & d = 4
\end{cases}
\end{equation}
and
\begin{equation}
\lim_{k \to \infty} \left \vert k \partial_k v \left( s \right) \right \vert \le \begin{cases}
0 & d < 4 \\
\sum_{n = 2}^\infty Y_{2n} \frac{s^{2n}}{\left( 2n \right)!} & d = 4 \, ,
\end{cases}
\end{equation}
for $Y_{2n} = B_{2n}^{1,1} + \left \vert 2 + \left( 2 - d \right) \left( n - 1 \right) \right \vert B_{2n}^{0,1}$.
In particular no definite statement is obtained by these methods for $d > 4$.
However, the $\kappa_4$ contribution to $v$ may be calculated explicitly:
\begin{equation}
\begin{aligned}
k \partial_k &\frac{\kappa_4 \left( 0, 0, 0 \right)}{k^{4-d}} = \lambda \exp \left[ - \frac{ \left \vert m \right \vert}{k} \right] \\
& \times \left( \left( 4 - d \right) \left( \frac{\left \vert m \right \vert}{k} \right)^{4 - d}
+  \left( \frac{\left \vert m \right \vert}{k} \right)^{5 - d}
\right) \, .
\end{aligned}
\end{equation}
Hence,
\begin{equation}
\lim_{k \to \infty} k \partial_k \frac{\kappa_4 \left( 0, 0, 0 \right)}{k^{4-d}}
=
\begin{cases}
0 & d \le 4 \\
-\infty & \text{otherwise} \, .
\end{cases}
\end{equation}
so that for $d = 4$ the beta function of the quartic term of the dimensionless potential vanishes in the limit of $k \to \infty$.

In $d < 4$ we see that the dimensionless potential as well as its flow vanishes in the large $k$ limit owing to the fact that $\kappa_2$ and $\kappa_4$ are bounded and have positive mass dimensions.
Hence, in comparison with the Wilson-Fisher fixed point the constructed solution features a completely different phenomenology.

In the case of $d = 4$ the dimensionless potential is obviously non-vanishing while the fate of its flow in the limit of $k \to \infty$ is unclear.
A numerical analysis showed that the coefficients $Y_{2n}$ with $d = 4$ grow so fast that a zero radius of convergence is probable.
Thus, we do not obtain a useful estimate of $\lim_{k \to \infty} \left \vert k \partial_k v \left( s \right) \right \vert$ in this case.
\section{Possible Applications}
The iterative scheme presented in section \ref{sec:AConstructiveSolution} provides a systematic approach to producing exact solutions to the expanded Wetterich equation.
It may be straightforwardly adapted to fermionic fields as well as to models with multiple fields with few modifications.
As such the approach is extremely general and can be applied in many situations.
Obvious candidates are theories in which approximation schemes have produced a set of one-particle irreducible correlation functions such as propagators or flows of lower-order vertices.
Such approximations may e.g have been produced by expansion schemes or lattice computations and are not limited to analytic input but can just as well be numerical.
In the case where such quantities have only been calculated at $k = 0$ without a regulator, a renormalisation group flow may be imposed through a suitable interpolation between the given result and some initial conditions along with a regulator.
The approach is then to construct operators $\rho_n$ that are compatible with boundary or regularity conditions and study features such as relevance and irrelevance of resulting higher order-operators along the renormalisation group flow.
Likewise, proposed flows of lower-order correlators may be scrutinised and possibly dismissed if the flow of the higher-order correlators proves to be singular or fails to have correct asymptotics.

This constitutes a new bootstrap strategy to explore exact properties of the theory space using the functional RG.
\section{Conclusions}
It has been demonstrated that a Euclidean invariant exact solution to equation \ref{eq:WetterichExpandedWithLambdaBars} satisfying the boundary conditions \ref{eq:initialConditions} exists and may be constructed as outlined in section \ref{sec:SolvingTheFlowEquations}.
Furthermore, explicit bounds on the flow as given by equation \ref{eq:Kappa2nSupBound} and more generally by the methods applied in appendices \ref{apx:ProofKappa2IterationConverges} and \ref{apx:BigProof} may be utilized to approximate the flow of any given correlation function to arbitrary precision.
By construction, the mass and the quartic coupling only undergo finite renormalisations during the flow from $k \to \infty$ to $k \to 0$.
Thus, the theory in the latter limit does not correspond to the known $\phi^4_3$ result which requires infinite renormalisations.
This raises the question of how to determine the physically correct boundary conditions in the large $k$ limit which is of course intimately connected with the physically appropriate choice of classical action $S_\Lambda$.
Conversely, one may ask how a given renormalisation group flow determines $S_\Lambda$ which precisely amounts to the reconstruction problem\cite{src:Manrique2009}.
With $S_\Lambda$ being unknown in this case, it is unclear whether $\lim_{k \to 0} \Gamma_k$ is independent of the choice of renormalisation scheme.
In particular it was demonstrated that the flow was not uniquely determined by $\lim_{k \to \infty} \Gamma_k$.
Hence, it may be expected that there is a yet to be uncovered connection between exact solutions to the flow equations and a possibly unique physical one.

The given solution was obtained through a very straightforward construction procedure that essentially enables the extrapolation of higher-order correlation functions from a set of lower-order ones.
Though these extrapolations should not be expected to be unique, one may hope that their asymptotic behaviour for small and large values of $k$ are strongly constrained.
Such constraints can then reveal lots of structure of the higher correlators.
In particular, the construction principle may be extended to models with multiple scalar fields as well as fermions without gauge symmetries.
Applying similar choices of $\rho_n$ operators to systems truncated at finite $n \in \mathbb{N}$ may then give hints for or against the applicability of the truncations in use and possibly even enable the explicit calculation of uncertainties.
\acknowledgments
I wish to thank Holger Gies for fruitful discussions about this paper as well as Abdol Sabor Salek for providing the induction hypothesis stated in equation \ref{eq:PropagatorNthDerivative}.

This work has been funded by the Deutsche Forschungsgemeinschaft (DFG) under Grant Nos. 398579334 (Gi328/9-1) and 406116891 within the Research Training Group RTG 2522/1.
\appendix
\section{Proof of the Derivative Identity for the Propagator}
\label{apx:ProofPropagatorDerivative}
Before stepping into the induction proof, note that equation \ref{eq:PropagatorNthDerivative} corresponds to equation \ref{eq:PropagatorFirstDerivative} for $n = 1$.
In order to further shorten notation, let us write
\begin{equation}
\left \langle c \right \rangle
=
\left \langle c_1, ..., c_l \right \rangle
=
A \circ \prod_{l = 1}^{\# c} \left( \Differential^{c_l} \Gamma_k^{(2)} \circ A \right)
\end{equation}
for all $l \in \mathbb{N}$ and any multi-index $c \in \mathbb{N}^l$.
Furthermore, define two operations on such multi-indices.
First, let
\begin{equation}
\begin{aligned}
s_j &: \mathbb{N}^l \to \mathbb{N}^l \\
\left( n_1, ..., n_l \right) &\mapsto \left( n_1, ..., n_{j-1}, 1 + n_j, n_{j+1}, ..., n_l \right)
\end{aligned}
\end{equation}
for any $j \in \mathbb{N}$ with $j \le l$ and second,
\begin{equation}
\begin{aligned}
t_j &: \mathbb{N}^l \to \mathbb{N}^{l+1} \\
\left( n_1, ..., n_l \right) &\mapsto \left( n_1, ..., n_{j-1}, 1, n_j, ..., n_l \right)
\end{aligned}
\end{equation}
for any $j \in \mathbb{N}$ with $j \le l + 1$.
Assuming the validity of equation \ref{eq:PropagatorNthDerivative} for a fixed $n \in \mathbb{N}$, we obtain
\begin{equation}
\begin{aligned}
\Differential^{n+1} A &=
\sum_{c \in \CombinationSet{n}} \left( -1 \right)^{1 + \# c} \frac{n!}{c!}
\sum_{j=1}^{1 + \# c}
\left \langle t_j \left( c \right) \right \rangle \\
&\phantom{=} + \sum_{c \in \CombinationSet{n}} \left( -1 \right)^{\# c} \frac{n!}{c!}
\sum_{j=1}^{\# c}
\left \langle s_j \left( c \right) \right \rangle \, .
\end{aligned}
\end{equation}
It is apparent that $s_j$ and $t_j$ are both injective maps from $\CombinationSet{n}$ to $\CombinationSet{n+1}$ for all possible $j$.
Thus, we may equally well sum over $\CombinationSet{n+1}$ instead of $\CombinationSet{n}$ giving
\begin{equation}
\begin{aligned}
\Differential^{n+1} A &=
\sum_{c \in \CombinationSet{n + 1}} \left( -1 \right)^{\# c} \frac{n!}{c!}
\sum_{\substack{j = 1 \\ c_j = 1}}^{\# c}
\left \langle c \right \rangle \\
&\phantom{=} + \sum_{c \in \CombinationSet{n + 1}} \left( -1 \right)^{\# c} \frac{n!}{c!}
\sum_{\substack{j = 1 \\ c_j \neq 1}}^{\# c}
c_j \left \langle c \right \rangle \, ,
\end{aligned}
\end{equation}
where it is now obvious that
\begin{equation}
\begin{aligned}
\Differential^{n+1} A &=
\sum_{c \in \CombinationSet{n + 1}} \left( -1 \right)^{\# c} \frac{n!}{c!}
\sum_{j = 1}^{\# c}
c_j \left \langle c \right \rangle \\
&=
\sum_{c \in \CombinationSet{n + 1}} \left( -1 \right)^{\# c} \frac{\left(n + 1 \right)!}{c!} \left \langle c \right \rangle \, .
\end{aligned}
\end{equation}
This proves equation \ref{eq:PropagatorNthDerivative}.
\section{Proof that \texorpdfstring{$I_n \circ \rho_n = \mathrm{id}$}{In after rhon = id}}
\label{apx:ProofRhonIsRightInverse}
Let us fix a real $\SymmetricGroupN{n-1}^*$-symmetric function $g$ on $\left( \mathbb{R}^d \right)^{n-1}$ and compute $I_n \rho_n g$.
In order to facilitate the proof, let us split $\rho_n g$ into the following parts defined by restricting the sum over $J$ in equation \ref{eq:rhonDefinition}:
\begin{itemize}
\item $\rho_n^1 g$ where $J$ contains no index $\ge n$
\item $\rho_n^2 g$ where $J$ contains precisely one index $\ge n$
\item $\rho_n^3 g$ where $J$ contains precisely two indices $\ge n$
\end{itemize}
Then, $I_n \rho_n g = I_n \rho_n^1 g + I_n \rho_n^2 g + I_n \rho_n^3 g$ by the linearity of $I_n$\footnote{This splitting makes sense, because $I_n$ is also defined for non-$\SymmetricGroupN{n+1}$-symmetric functions.}.
Hence, it suffices to analyse the three parts individually:
The first part becomes
\begin{equation}
\begin{aligned}
\big( I_n \rho_n^1 &g \big) \left( p_1, ..., p_{n-1} \right) =
\sum_{J \subseteq \left \{ 0, ..., n - 1 \right \}}
\sum_{l = 0}^{\left \lfloor \frac{n - 1 - \# J}{2} \right \rfloor} \\
&\frac{\alpha^{n}_{\# J, l}}{\left( \int_{\mathbb{R}^d} K \right)^{n - \# J - l}} 
\int_{\left( \mathbb{R}^d \right)^{n - 1 - \# J - l}}
\int_{\mathbb{R}^d} K \left( q \right)\\
&g \left( p_J, -s_1, s_1, ..., -s_l, s_l, t_1, ..., t_{n - 1 - \# J -2l} \right) \\
&K \left( s_1 \right) ... K \left( s_l \right)
K \left( t_1 \right) ... K \left( t_{n - 1 - \# J - 2l} \right) \\
&\mathrm{d} q \, \mathrm{d} s_{...} \mathrm{d} t_{...} \, ,
\end{aligned}
\end{equation}
where $p_J$ contains neither $q$ nor $-q$ allowing the evaluation of the $q$ integral.
Thus,
\begin{equation}
\label{eq:InRhonPart1}
\begin{aligned}
\big( I_n \rho_n^1 &g \big) \left( p_1, ..., p_{n-1} \right) =
\sum_{J \subseteq \left \{ 0, ..., n - 1 \right \}}
\sum_{l = 0}^{\left \lfloor \frac{n - 1 - \# J}{2} \right \rfloor} \\
&\frac{\alpha^{n}_{\# J, l}}{\left( \int_{\mathbb{R}^d} K \right)^{n - 1 - \# J - l}} 
\int_{\left( \mathbb{R}^d \right)^{n - 1 - \# J - l}} \\
&g \left( p_J, -s_1, s_1, ..., -s_l, s_l, t_1, ..., t_{n - 1 - \# J -2l} \right) \\
&K \left( s_1 \right) ... K \left( s_l \right)
K \left( t_1 \right) ... K \left( t_{n - 1 - \# J - 2l} \right) \\
&\mathrm{d} s_{...} \mathrm{d} t_{...} \, .
\end{aligned}
\end{equation}
In the second part $p_J$ contains either $q$ or $-q$.
But since $K \left( q \right) = K \left( -q \right)$ by equation \ref{eq:RegulatorIsSym1AstSymmetric} both contributions are identical.
Removing the index $n$ or $n+1$ respectively from $J$ and inserting $q$ explicitly then leads to
\begin{equation}
\begin{aligned}
\big( I_n &\rho_n^2 g \big) \left( p_1, ..., p_{n-1} \right) =
2 \sum_{J \subseteq \left \{ 0, ..., n - 1 \right \}}
\sum_{l = 0}^{\left \lfloor \frac{n - 2 - \# J}{2} \right \rfloor} \\
&\frac{\alpha^{n}_{\# J + 1, l}}{\left( \int_{\mathbb{R}^d} K \right)^{n - 1 - \# J - l}} 
\int_{\left( \mathbb{R}^d \right)^{n - 2 - \# J - l}}
\int_{\mathbb{R}^d} K \left( q \right) \\
&g \left( p_J, q, -s_1, s_1, ..., -s_l, s_l, t_1, ..., t_{n - 2 - \# J -2l} \right) \\
&K \left( s_1 \right) ... K \left( s_l \right)
K \left( t_1 \right) ... K \left( t_{n - 2 - \# J - 2l} \right) \\
&\mathrm{d} q \, \mathrm{d} s_{...} \mathrm{d} t_{...} \, ,
\end{aligned}
\end{equation}
where the factor of $2$ comes from the two possibilities of picking either $n$ or $n+1$.
Relabelling $q$ to $t_{n-1-\# J - 2l}$ simplifies this part to
\begin{equation}
\label{eq:InRhonPart2}
\begin{aligned}
\big( I_n \rho_n^2 &g \big) \left( p_1, ..., p_{n-1} \right) =
\sum_{J \subseteq \left \{ 0, ..., n - 1 \right \}}
\sum_{l = 0}^{\left \lfloor \frac{n - 2 - \# J}{2} \right \rfloor} \\
&\frac{2 \alpha^{n}_{\# J + 1, l}}{\left( \int_{\mathbb{R}^d} K \right)^{n - 1 - \# J - l}} 
\int_{\left( \mathbb{R}^d \right)^{n - 1 - \# J - l}} \\
&g \left( p_J, -s_1, s_1, ..., -s_l, s_l, t_1, ..., t_{n - 1 - \# J -2l} \right) \\
&K \left( s_1 \right) ... K \left( s_l \right)
K \left( t_1 \right) ... K \left( t_{n - 1 - \# J - 2l} \right) \\
&\mathrm{d} s_{...} \mathrm{d} t_{...} \, ,
\end{aligned}
\end{equation}
where the similarity to equation \ref{eq:InRhonPart1} is immediate.
In the third part $J$ contains both $n$ and $n+1$ corresponding to $p_J$ containing both $q$ and $-q$.
Removing these indices from $J$, one obtains
\begin{equation}
\begin{aligned}
\big( &I_n \rho_n^3 g \big) \left( p_1, ..., p_{n-1} \right) =
\sum_{J \subseteq \left \{ 0, ..., n - 1 \right \}}
\sum_{l = 0}^{\left \lfloor \frac{n - 3 - \# J}{2} \right \rfloor} \\
&\frac{\alpha^{n}_{\# J + 2, l}}{\left( \int_{\mathbb{R}^d} K \right)^{n - 2 - \# J - l}} 
\int_{\left( \mathbb{R}^d \right)^{n - 3 - \# J - l}}
\int_{\mathbb{R}^d} K \left( q \right)\\
&g \left( p_J, -q, q, -s_1, s_1, ..., -s_l, s_l, t_1, ..., t_{n - 3 - \# J -2l} \right) \\
&K \left( s_1 \right) ... K \left( s_l \right)
K \left( t_1 \right) ... K \left( t_{n - 3 - \# J - 2l} \right) \\
&\mathrm{d} q \, \mathrm{d} s_{...} \mathrm{d} t_{...} \, .
\end{aligned}
\end{equation}
Relabelling $q$ to $s_{l+1}$ and shifting the index $l$ by $1$ leads to
\begin{equation}
\label{eq:InRhonPart3}
\begin{aligned}
\big( I_n \rho_n^3 &g \big) \left( p_1, ..., p_{n-1} \right) =
\sum_{J \subseteq \left \{ 0, ..., n - 1 \right \}}
\sum_{l = 1}^{\left \lfloor \frac{n - 1 - \# J}{2} \right \rfloor} \\
&\frac{\alpha^{n}_{\# J + 2, l - 1}}{\left( \int_{\mathbb{R}^d} K \right)^{n - 1 - \# J - l}} 
\int_{\left( \mathbb{R}^d \right)^{n - 1 - \# J - l}} \\
&g \left( p_J, -s_1, s_1, ..., -s_l, s_l, t_1, ..., t_{n - 1 - \# J - 2l} \right) \\
&K \left( s_1 \right) ... K \left( s_l \right)
K \left( t_1 \right) ... K \left( t_{n - 1 - \# J - 2l} \right) \\
&\mathrm{d} s_{...} \mathrm{d} t_{...} \, .
\end{aligned}
\end{equation}
It is now straightforward to add up the parts in equations \ref{eq:InRhonPart1}, \ref{eq:InRhonPart2} and \ref{eq:InRhonPart3}.
Furthermore, the coefficients $\alpha^n_{a,b}$ may be determined by demanding $I_{n} \rho_{n} g = g$ translating to
\begin{equation}
\begin{aligned}
&\qquad n \, \alpha^{n}_{n-1,0} = 1 \\
&\forall a \in \left \{ 0, ..., n - 4 \right \}, b \in \left \{ 1, ..., \left \lfloor \frac{n - 2 - a}{2} \right \rfloor \right \} : \\ &\qquad \alpha^{n}_{a,b} + 2 \alpha^{n}_{a+1, b} + \alpha^{n}_{a+2,b-1} = 0 \\
&\forall a \in \left \{ 0, ..., n - 3 \right \}, n - a \text{ odd} : \\
&\qquad \alpha^n_{a,\left( n-1-a \right) / 2} + \alpha^n_{a+2,\left( n-3-a \right) / 2} = 0 \\
&\forall a \in \left \{ 0, ..., n - 2 \right \} : \\
&\qquad \alpha^{n}_{a,0} + 2 \alpha^{n}_{a+1,0} = 0 \, .
\end{aligned}
\end{equation}
Here the first factor of $n$ comes from the $n$ different subsets of $\left \{ 0, ..., n-1 \right \}$ of length $n-1$.
All these subsets give the same contribution to $I_n \rho_n g$ due to the $\SymmetricGroupN{n-1}^*$ symmetry of $g$.
As may easily be verified, equation \ref{eq:alphanabDefinition} solves these recursion relations.
Furthermore, this solution is unique because all $\alpha^n_{a,b}$ for $n \in \mathbb{N}$ and $a, b \in \mathbb{N}_0$ with $a + 2b \le n - 1$ are uniquely determined by the values of $\alpha^n_{n-1,0}$.
\section{Existence Proof of \texorpdfstring{$\kappa_2$}{kappa2}}
\label{apx:ProofKappa2IterationConverges}
Let $\kappa_2^1 \left( p; k \right) = m^2 + \left \Vert p \right \Vert^2$ and for any $n \in \mathbb{N}$ define
\begin{equation}
\begin{aligned}
\label{eq:kappa2Iteration}
\kappa_2^{n+1} &\left( p; k \right) =
\kappa_2^1 \left( p \right) +
\frac{1}{2 \left( 2 \pi \right)^d}
\int_k^\infty \int_{\mathbb{R}^d} \\
&\frac{\partial_{k'} r \left( q, k' \right)}{\left[ \kappa_2^n \left( q; k' \right) + \bar{r} \left( q, k' \right) \right]^2}
\kappa_{4} \left(p, - q, q; k' \right)
\mathrm{d} q \, \mathrm{d} k'
\end{aligned}
\end{equation}
which obviously satisfies the boundary condition \ref{eq:initialConditions} if the integrals are finite.
The first thing to note is that $\kappa_2^n \ge \kappa_2^1$ for all $n \in \mathbb{N}$, since $\kappa_2^1 > 0$, $\kappa_4 \ge 0$ and by equation \ref{eq:RegulatorDefinition} the regulator contribution is positive.
Thus, by equation \ref{eq:inversePropagatorEstimate}
\begin{equation}
\frac{1}{\left[ \kappa_2^n \left( q \right) + \bar{r} \left( q \right) \right]^2}
\le
\frac{1}{\left[ m^2 + k^2 \right]^2}
\end{equation}
which together with
\begin{equation}
\partial_k \bar{r} \left( q \right)
=
\frac{\left \Vert q \right \Vert^4}{k^3 \left( \cosh \left[ \frac{\left \Vert q \right \Vert^2}{k^2} \right] - 1 \right)} \\
\le
2 k
\end{equation}
leads to
\begin{equation}
\frac{\partial_k \bar{r} \left( q \right)}{\left[ \kappa_2^n \left( q \right) + \bar{r} \left( q \right) \right]^2}
\le
\frac{2 k}{\left[ m^2 + k^2 \right]^2} \, .
\end{equation}
Inserting this into the recursion relation \ref{eq:kappa2Iteration} leads to
\begin{equation}
\begin{aligned}
\kappa_2^n \left( p \right) &\le
\kappa_2^1 \left( p \right) + \frac{\left( 2 \pi \right)^{-d} \lambda}{\left \vert m \right \vert^{d-4}}
\int_k^\infty  \int_{\mathbb{R}^d}
\frac{k'}{\left[ m^2 + k'^2 \right]^2} \\
&\phantom{\le}
\times
\exp \left[ - \frac{2 \left \Vert p \right \Vert^d + 2 \left \Vert q \right \Vert^d + \left \vert m \right \vert^d}{k' \left \vert m \right \vert^{d-1}} \right]
\mathrm{d} q \, \mathrm{d} k' \\
&=
\kappa_2^1 \left( p \right) + \left( 2 \pi \right)^{-d} \frac{s_{d-1}}{2 d} \lambda \left \vert m \right \vert^3
\int_k^\infty \\
&\phantom{=} \times
\frac{k'^2}{\left[ m^2 + k'^2 \right]^2}
\exp \left[ - \frac{2 \left \Vert p \right \Vert^d + \left \vert m \right \vert^d}{k' \left \vert m \right \vert^{d-1}} \right]
\mathrm{d} k' \, .
\end{aligned}
\end{equation}
where $s_n$ denotes the surface area of the unit $n$-sphere.
Estimating the exponential by $1$ and extending the integral to $\left[ 0, \infty \right)$ immediately gives the result
\begin{equation}
\begin{aligned}
\kappa_2^n \left( p \right) &\le
\kappa_2^1 \left( p \right) + \left( 2 \pi \right)^{-d} \frac{s_{d-1}}{2 d} \lambda \left \vert m \right \vert^3 \\
&\phantom{\le}
\times \int_0^\infty
\frac{k'^2}{\left[ m^2 + k'^2 \right]^2}
\mathrm{d} k' \\
&\le
\kappa_2^1 \left( p \right) + \left( 2 \pi \right)^{-d} \pi \frac{s_{d-1}}{8 d} \lambda m^2
\end{aligned}
\end{equation}
which in a slightly more compact form reads
\begin{equation}
\left \Vert \kappa_2^n - \kappa_2^1 \right \Vert_{L^\infty} \le \frac{ \left( 2 \pi \right)^{-d} \pi s_{d-1}}{8 d} \lambda m^2
:=
t_d \lambda m^2
\end{equation}
for all $n \in \mathbb{N}$.
Note that the numerical factor $t_d$ in front of $\lambda m^2$ is rather small: It is $1/8$ for $d = 1$ and goes to zero rather rapidly for larger values of $d$.

We shall now show that the mapping $\kappa_2^n \mapsto \kappa_2^{n+1}$ given by equation \ref{eq:kappa2Iteration} actually is a contraction for values of $\lambda$ not being too large.
To this end, note that
\begin{equation}
\begin{aligned}
\frac{\kappa_2^n \left( q \right) + \bar{r} \left( q \right)}{\kappa_2^1 \left( q \right) + \bar{r} \left( q \right)} &=
\frac{\kappa_2^1 \left( q \right) + \bar{r} \left( q \right)}{\kappa_2^1 \left( q \right) + \bar{r} \left( q \right)} +
\frac{\kappa_2^n \left( q \right) - \kappa_2^1 \left( q \right)}{\kappa_2^1 \left( q \right) + \bar{r} \left( q \right)} \\
&\le
1 + t_d \lambda \frac{m^2}{\kappa_2^1 \left( q \right) + \bar{r} \left( q \right)} \\
&\le
1 + t_d \lambda
\end{aligned}
\end{equation}
and hence
\begin{equation}
\begin{aligned}
&\left \vert
\left[ \kappa_2^{n+1} \left( q \right) + \bar{r} \left( q \right) \right]^{-2}
-
\left[ \kappa_2^n \left( q \right) + \bar{r} \left( q \right) \right]^{-2}
\right \vert \\
&=
\frac{\left[ 2 \bar{r} \left( q \right) + \kappa_2^n \left( q \right) + \kappa_2^{n+1} \left( q \right) \right] \left \vert \kappa_2^n \left( p \right) - \kappa_2^{n+1} \left( q \right) \right \vert}{\left[ \kappa_2^{n+1} \left( q \right) + \bar{r} \left( q \right) \right]^2 \left[ \kappa_2^n \left( q \right) + \bar{r} \left( q \right) \right]^2 } \\
&\le
\frac{\left[ 2 \bar{r} \left( q \right) + \kappa_2^n \left( q \right) + \kappa_2^{n+1} \left( q \right) \right] \left \vert \kappa_2^n \left( p \right) - \kappa_2^{n+1} \left( q \right) \right \vert}{\left[ \kappa_2^1 \left( q \right) + \bar{r} \left( q \right) \right]^4} \\
&\le
2 \left( 1 + t_d \lambda \right)
\frac{\left \vert \kappa_2^n \left( q \right) - \kappa_2^{n+1} \left( q \right) \right \vert}{\left[ \kappa_2^1 \left( q \right) + \bar{r} \left( q \right) \right]^3} \\
&\le
2 \frac{1 + t_d \lambda}{m^2}
\frac{\left \Vert \kappa_2^n - \kappa_2^{n+1} \right \Vert_{L^\infty}}{\left[ \kappa_2^1 \left( q \right) + \bar{r} \left( q \right) \right]^2} \, .
\end{aligned}
\end{equation}
Using this estimate to compare two successive iterates one finally arrives at
\begin{equation}
\begin{aligned}
\big \vert &\kappa_2^{n+2} \left( p \right) - \kappa_2^{n+1} \left( p \right) \big \vert \\
&\le
2 \frac{1 + t_d \lambda}{m^2}
\frac{\left \Vert \kappa_2^n - \kappa_2^{n+1} \right \Vert_{L^\infty}}{2 \left( 2 \pi \right)^d}
\int_k^\infty
\int_{\mathbb{R}^d} \\
&\phantom{\le} \times
\frac{\partial_{k'} r \left( q; k' \right)}{\left[ \kappa_2^1 \left( q \right) + \bar{r} \left( q \right) \right]^2}
\kappa_{4} \left(p, - q, q; k' \right)
\mathrm{d} q \, 
\mathrm{d} k' \\
&\le
2 \frac{1 + t_d \lambda}{m^2}
\left \Vert \kappa_2^n - \kappa_2^{n+1} \right \Vert_{L^\infty}
\left \Vert \kappa_2^2 - \kappa_2^1 \right \Vert_{L^\infty} \\
&\le
2 \left( 1 + t_d \lambda \right) t_d \lambda
\left \Vert \kappa_2^{n+1} - \kappa_2^n \right \Vert_{L^\infty} \, ,
\end{aligned}
\end{equation}
or for short
\begin{equation}
\begin{aligned}
\big \Vert \kappa_2^{n+2} - &\kappa_2^{n+1} \big \Vert_{L^\infty}
\le \\
&2 \left( 1 + t_d \lambda \right) t_d \lambda
\left \Vert \kappa_2^{n+1} - \kappa_2^n \right \Vert_{L^\infty} \, .
\end{aligned}
\end{equation}
The factor in front is smaller than one whenever
\begin{equation}
\label{eq:lambdaBoundImplicit}
0 \le \lambda < \frac{\sqrt{3} - 1}{2 t_d} \, ,
\end{equation}
or equivalently equation \ref{eq:lambdaBoundFromIteration} is satisfied.
The upper bound is a function that grows rather rapidly starting at a value of approximately $2.93$ for $d=1$.
From now on, we assume $\lambda$ to satisfy inequality \ref{eq:lambdaBoundImplicit}.
Thus, by the completeness of $L^\infty \left( \mathbb{R}^d \right)$ we have proven the convergence of the sequence $\left( p \mapsto \kappa_2^n \left( p; k \right) \right)_{n \in \mathbb{N}}$ to some $p \mapsto \kappa_2 \left( p; k \right)$ in $L^\infty \left( \mathbb{R}^d \right)$ for all $k \in \left[ 0, \infty \right)$.
Also, $\kappa_2$ has to be a fixed point of the iteration map such that equation \ref{eq:Kappa2isAFixedPointOfIteration} is satisfied where the right hand side is continuous with respect to $k$, since the integrand is non-singular for all $k' \ge 0$.
Thus, $\kappa_2$ is also $k$-continuous on $\left[ 0, \infty \right)$ as well.
But then the right-hand side is differentiable with respect to $k$ on all of $\left[ 0, \infty \right)$, so that
\begin{equation}
\begin{aligned}
\partial_k &\kappa_2 \left( p \right) =
- \frac{1}{2} \left( 2 \pi \right)^{-d} \\
&\times
\int_{\mathbb{R}^d} \frac{\partial_k r \left( q \right)}{\left[ \kappa_2 \left( q \right) + r \left( q \right) \right]^2}
\kappa_{4} \left(p, - q, q \right)
\mathrm{d} q
\end{aligned}
\end{equation}
for all $k \in \mathbb{R}_{\ge 0}$.
Hence, $\kappa_2$ satisfies the flow equation.
Furthermore, the right-hand side is obviously $k$-differentiable so that $\partial_k^2 \kappa_2$ may be expressed through $\kappa_2$ and $\partial_k \kappa_2$.
Hence, $\partial_k^2 \kappa_2$ is again $k$-differentiable.
Iterating this argument then shows that $\kappa_2$ is smooth with respect to $k$.
The $p$-smoothness of $\kappa_2$ is immediate from equation \ref{eq:Kappa2isAFixedPointOfIteration} by the regularity of $\kappa_4$.
For the $\mathcal{O} \left( d \right)$-invariance of $\kappa_2$, note that $\kappa_4$ and $\bar{r}$ as well as $\kappa_2^1$ are $\mathcal{O} \left( d \right)$-invariant.
Thus, by equation \ref{eq:kappa2Iteration} each iterate $\kappa_2^n$ is also $\mathcal{O} \left( d \right)$-invariant.
Since the set of all $\mathcal{O} \left( d \right)$-invariant functions in $L^\infty \left( \mathbb{R}^d \right)$ is closed, the limit point $\kappa_2$ has to lie in this set as well.
\section{Bounding the Higher Correlators}
\label{apx:BigProof}
Let us assume that all higher correlators have been constructed by virtue of equation \ref{eq:Kappa2np2FromRightInverse}.
It then remains to find useful bounds ascertaining the correct UV limits as well as non-singular IR limits.
The key to this, is a proper estimate for the $k$-derivatives of $\kappa_2$.
Before we can produce such estimates, we shall need corresponding ones for $\kappa_4$ and $\bar{r}$.
Let us begin with the regulator for which we have the relation
\begin{equation}
\partial_k \bar{r} \left( q \right)
=
\frac{2}{k^3} \bar{r} \left( q \right) \left[ \left \Vert q \right \Vert^2 + \bar{r} \left( q \right) \right]
\end{equation}
that may easily be derived from equation \ref{eq:RegulatorDefinition}.
It hints at the following identity for all $l \in \mathbb{N}_0$ and some constants $\beta^l_{a,b} \in \mathbb{R}$:
\begin{equation}
\begin{aligned}
\label{eq:RegulatorKDerivative}
\partial_k^l &\bar{r} \left( q \right)
= \\
&\sum_{a = 1}^{l+1} \sum_{b = 0}^{l + 1 - a} \beta^l_{a,b} k^{2-l-2a-2b} \bar{r} \left( q \right)^a \left[ \left \Vert q \right \Vert^2 + \bar{r} \left( q \right) \right]^b
\end{aligned}
\end{equation}
which can straightforwardly be proved by induction.
The constants $\beta^l_{a,b}$ are recursively defined by
\begin{equation}
\begin{aligned}
\beta^{l+1}_{a,b} &= \left( 2 - l - 2a -2b \right) \beta^l_{a,b} \\
&\phantom{=} + 2 a \beta^l_{a,b-1} + 2 b \beta^l_{a-1,b} \\
\beta^0_{a,b} &= \begin{cases}
1 & a = 1, b = 0 \, , \\
0 & \text{otherwise}
\end{cases}
\end{aligned}
\end{equation}
for all $l \in \mathbb{N}_0$ and $a, b \in \mathbb{Z}$.
The next theorem will allow to find an estimate for such expressions.

\begin{theorem}
\label{thm:RegulatorPowersEstimate}
Let $a \in \mathbb{N}$ and $b \in \mathbb{N}_0$.
Then,
\begin{equation}
\sup_{q \in \mathbb{R}^d} \left \vert \bar{r} \left( q \right)^a \left[ \left \Vert q \right \Vert^2 + \bar{r} \left( q \right) \right]^b \right \vert \le k^{2 \left( a + b \right)} \left( 1 + \frac{b}{a} \right)^b \, .
\end{equation}
\end{theorem}
\begin{proof}
For $b = 0$ the statement is obvious since $\bar{r} \left( q \right) \le k^2$.
Hence, let us assume that $b \in \mathbb{N}$.
Since $\bar{r} \left( q \right)^a \left[ \left \Vert q \right \Vert^2 + \bar{r} \left( q \right) \right]^b$ is actually a smooth function of $\left \Vert q \right \Vert^2$, we may look for local extrema by differentiating with respect to $\left \Vert q \right \Vert^2$.
Then, a necessary condition for $\left \Vert q \right \Vert^2$ at a maximum is
\begin{equation}
\begin{aligned}
a \left[ \left \Vert q \right \Vert^2 + \bar{r} \left( q \right) \right] &\partial_{\left \Vert q \right \Vert^2} \bar{r} \left( q \right) \\
&+
b \bar{r} \left( q \right) \left[ 1 + \partial_{\left \Vert q \right \Vert^2} \bar{r} \left( q \right) \right]
= 0 \, .
\end{aligned}
\end{equation}
Now, note that the exponential regulator also admits the following simple identity for $q \neq 0$:
\begin{equation}
\label{eq:q2DerivativeOfRegulator}
\partial_{\left \Vert q \right \Vert^2} \bar{r} \left( q \right) = \frac{\bar{r} \left( q \right)}{\left \Vert q \right \Vert^2} \left[ 1 - \frac{1}{k^2} \left( \left \Vert q \right \Vert^2 + \bar{r} \left( q \right) \right) \right]
\end{equation}
which inserted into the previous equation gives us the equivalent condition
\begin{equation}
\label{eq:RegulatorPowersExtremePointConditionIntermediate}
1 = \frac{\bar{r} \left( q \right)}{k^2} + \frac{a}{a+b} \frac{\left \Vert q \right \Vert^2}{k^2}
\end{equation}
after some simple algebra.
We perform a change of variables to $y = \frac{\left \Vert q \right \Vert^2}{k^2}$ and obtain
\begin{equation}
\label{eq:RegulatorPowersExtremePointCondition}
\exp y = 1 + \frac{\left( a + b \right) y}{a + b - a y}
\end{equation}
as a further equivalent expression for the extremality even including the case $q = 0$.
Note, that the apparently excluded case $a y = a + b$ is not relevant, since it does not solve the derivative test as is obvious from equation \ref{eq:RegulatorPowersExtremePointConditionIntermediate}.
Furthermore,
\begin{equation}
\frac{\partial}{\partial y} \left( 1 + \frac{\left( a + b \right) y}{a + b - a y} \right) = \left( \frac{a + b}{a + b - a y} \right)^2 > 0 \, ,
\end{equation}
such that for $a y > a + b$ we have the right-hand side of equation \ref{eq:RegulatorPowersExtremePointCondition} being monotonically increasing with $y$ and
\begin{equation}
\lim_{y \to \infty} \left( 1 + \frac{\left( a + b \right) y}{a + b - a y} \right)
=
1 - \frac{a+b}{a}
= - \frac{b}{a}
< 0 \, ,
\end{equation}
spoiling equation \ref{eq:RegulatorPowersExtremePointCondition}.
Thus, we conclude that all extrema lie in the interval $y \in \left[ 0, \frac{a+b}{a} \right)$.
But then, at a maximum we have
\begin{equation}
\begin{aligned}
\bar{r} &\left( q \right)^a \left[ \left \Vert q \right \Vert^2 + \bar{r} \left( q \right) \right]^b \\
&= \left( \frac{\left \Vert q \right \Vert^2}{\exp y - 1} \right)^a \left( \left \Vert q \right \Vert^2 + \frac{\left \Vert q \right \Vert^2}{\exp y - 1} \right)^b \\
&= k^{2 \left( a + b \right)} \left( \frac{y}{\exp y - 1} \right)^a \left( y + \frac{y}{\exp y - 1} \right)^b \\
&= k^{2 \left( a + b \right)} \left( \frac{y}{\frac{\left( a + b \right) y}{a + b - a y}} \right)^a \left( y + \frac{y}{\frac{\left( a + b \right) y}{a + b - a y}} \right)^b \\
&= k^{2 \left( a + b \right)} \left( 1 - \frac{a}{a+b} y \right)^a \left( 1 + \frac{b}{a+b} y \right)^b \\
&\le
k^{2 \left( a + b \right)} \left( 1 + \frac{b}{a} \right)^b \, ,
\end{aligned}
\end{equation}
where we have used $y \in \left[ 0, \frac{a+b}{a} \right)$ in the last estimate.
\end{proof}
\begin{corollary}
Applying this estimate to the regulator derivatives, we obtain
\begin{equation}
\label{eq:RegulatorDerivativeEstimate}
\left \Vert \partial_k^n \bar{r} \right \Vert_{L^\infty}
\le
k^{2 - n} \sum_{a = 1}^n \sum_{b=0}^{n + 1 - a}
\left \vert \beta^n_{a, b} \right \vert \left( 1 + \frac{b}{a} \right)^b \, .
\end{equation}
Hence, there is a constant $R_n \ge 0$ such that
\begin{equation}
\left \Vert \partial_k^n \bar{r} \right \Vert_{L^\infty}
\le R_n k^{2 - n}
\end{equation}
for all $n \in \mathbb{N}_0$.
\end{corollary}
\begin{corollary}
\label{cor:KSupEstimate}
Applying the estimate to $K$ and employing equation \ref{eq:inversePropagatorEstimate} leads to
\begin{equation}
\left \Vert K \right \Vert_{L^\infty} \le R_1 \frac{k}{\left( k^2 + m^2 \right)^2} \, .
\end{equation}
\end{corollary}
Having obtained the appropriate estimates for the regulator, the next step towards the $\kappa_2$ estimates is to study $\kappa_4$.
\begin{theorem}
\label{thm:Kappa4SupIntEstimate}
For all $l \in \mathbb{N}_0$ there exist constants $A_4^l \ge 0$ such that
\begin{equation}
\sup_{p \in \mathbb{R}^d} \int_{\mathbb{R}^d} \left \vert
\partial_k^l \kappa_4 \left( p, q, -q \right)
\right \vert
\mathrm{d} q
\le
A_4^l \frac{\left \vert m \right \vert^3 k}{k^l + \left \vert m \right \vert^l} \, .
\end{equation}
\end{theorem}
\begin{proof}
As can easily be proved by induction, we have
\begin{equation}
\label{eq:Kappa4KDerivative}
\begin{aligned}
&\partial_k^l \kappa_4 \left( p, q, r \right)
=
\kappa_4 \left( p, q, r \right)
\sum_{a = 0}^{l}
\gamma^l_a \left \vert m \right \vert^{a - ad} k^{-a-l} \\
&\times \left( \left \Vert p \right \Vert^d + \left \Vert q \right \Vert^d + \left \Vert r \right \Vert^d + \left \Vert p + q + r \right \Vert^d + \left \vert m \right \vert^d \right)^a
\end{aligned}
\end{equation}
for all $l \in \mathbb{N}_0$, $p, q, r \in \mathbb{R}^d$. The constants $\gamma^l_a \in \mathbb{R}$ are determined by
\begin{equation}
\begin{aligned}
\gamma^{l+1}_a &= - \left( a + l \right) \gamma^l_a + \gamma^l_{a-1} \, ,\\
\gamma^0_a &= \begin{cases}
1 & a = 0 \, , \\
0 & \text{otherwise}
\end{cases}
\end{aligned}
\end{equation}
for all $a \in \mathbb{Z}$.
Expanding the above, we get
\begin{equation}
\begin{aligned}
\big \vert \partial_k^l &\kappa_4 \left( p, q, -q \right) \big \vert
\le
\lambda
\sum_{a = 0}^{l}
\sum_{b = 0}^a
\binom{a}{b}
2^b \left \vert \gamma^l_a \right \vert \\
&\times
\left \vert m \right \vert^{4 - d + a - ad} k^{-a-l}
\left( 2 \left \Vert p \right \Vert^d + \left \vert m \right \vert^d \right)^{a-b} \\
&\times
\left \Vert q \right \Vert^{b d}
\exp \left[ - \frac{2 \left \Vert p \right \Vert^d + 2 \left \Vert q \right \Vert^d + \left \vert m \right \vert^d}{k \left \vert m \right \vert^{d-1}} \right]
\end{aligned}
\end{equation}
which allows us to perform the $q$ integral, such that
\begin{equation}
\begin{aligned}
\int_{\mathbb{R}^d} &\left \vert \partial_k^l \kappa_4 \left( p, q, -q \right) \right \vert \mathrm{d} q
\le
\frac{s_{d-1}}{2 d} \lambda
\sum_{a = 0}^{l}
\sum_{b = 0}^a \\
&\times \binom{a}{b} b!
\left \vert \gamma^l_a \right \vert k^{1+b-a-l} \left( 2 \left \Vert p \right \Vert^d + \left \vert m \right \vert^d \right)^{a-b} \\
&\times
\left \vert m \right \vert^{3 + \left( b - a\right) \left( d - 1 \right)}
\exp \left[ - \frac{2 \left \Vert p \right \Vert^d + \left \vert m \right \vert^d}{k \left \vert m \right \vert^{d-1}} \right] \, .
\end{aligned}
\end{equation}
Let us again expand this, leading to
\begin{equation}
\begin{aligned}
\int_{\mathbb{R}^d} &\left \vert \partial_k^l \kappa_4 \left( p, q, -q \right) \right \vert \mathrm{d} q
\le
\frac{s_{d-1}}{2 d} \lambda
\sum_{a = 0}^{l}
\sum_{b = 0}^a
\sum_{c = 0}^{a-b} \\
&\times \binom{a}{b} \binom{a - b}{c} b! \, 2^c
\left \vert \gamma^l_a \right \vert \left \vert m \right \vert^{3 + a - b - c d} \\
&\times k^{1+b-a-l}
\left \Vert p \right \Vert^{c d}
\exp \left[ - \frac{2 \left \Vert p \right \Vert^d + \left \vert m \right \vert^d}{k \left \vert m \right \vert^{d-1}} \right] \, ,
\end{aligned}
\end{equation}
allowing us to produce the estimate
\begin{equation}
\label{eq:Kappa4SupIntEstimateIntermediate}
\begin{aligned}
&\sup_{p \in \mathbb{R}^d} \int_{\mathbb{R}^d} \left \vert \partial_k^l \kappa_4 \left( p, q, -q \right) \right \vert \mathrm{d} q
\le
\frac{s_{d-1}}{2 d} \lambda
\sum_{a = 0}^{l}
\sum_{b = 0}^a \\
&\; \times
\binom{a}{b} b! \left \vert \gamma^l_a \right \vert
\left \vert m \right \vert^{3+a-b} k^{1+b-a-l} \exp \left[ - \frac{\left \vert m \right \vert}{k} \right] \\
&+ \frac{s_{d-1}}{2 d} \lambda
\sum_{a = 0}^{l}
\sum_{b = 0}^a
\sum_{c = 1}^{a-b}
\binom{a}{b} \binom{a - b}{c} b! \left( \frac{c}{e} \right)^c
\left \vert \gamma^l_a \right \vert \\
&\; \times
\left \vert m \right \vert^{3 + a - b - c}
k^{1+b+c-a-l}
\exp \left[ - \frac{\left \vert m \right \vert}{k} \right] \, .
\end{aligned}
\end{equation}
For $l = 0$, the above reduces to
\begin{equation}
\begin{aligned}
\sup_{p \in \mathbb{R}^d} \int_{\mathbb{R}^d} &\left \vert \kappa_4 \left( p, q, -q \right) \right \vert \mathrm{d} q \\
&\le
\frac{s_{d-1}}{2 d} \lambda
\left \vert \gamma^0_0 \right \vert
\left \vert m \right \vert^3 k \exp \left[ - \frac{\left \vert m \right \vert}{k} \right] \\
&\le 
\frac{s_{d-1}}{2 d}
\left \vert \gamma^0_0 \right \vert
\left \vert m \right \vert^3 k \, ,
\end{aligned}
\end{equation}
which is precisely of the desired form.
For $l \in \mathbb{N}$ and all $a, b, c \in \mathbb{N}_0$ with $a - b - c \ge 0$ the following is valid:
\begin{equation}
\label{eq:ExponentialScreeningTrick}
\begin{aligned}
\exp &\left[ - \frac{\left \vert m \right \vert}{k} \right] \\
&\le
\left[
\frac{\left( \frac{\left \vert m \right \vert}{k} \right)^{a - b - c}}{\left(a-b-c \right)!} + \frac{\left( \frac{\left \vert m \right \vert}{k} \right)^{l + a - b - c}}{\left(l + a - b - c\right)!}
\right]^{-1} \\
&=
\frac{k^{l+a-b-c} \left \vert m \right \vert^{b + c - a}}{\frac{1}{\left(a-b-c \right)!} k^{l} + \frac{1}{\left(l + a - b - c\right)!}\left \vert m \right \vert^{l}} \\
&\le
\left(l + a - b - c\right)! \frac{k^{l+a-b-c} \left \vert m \right \vert^{b + c - a}}{k^{l} + \left \vert m \right \vert^{l}} \, .
\end{aligned}
\end{equation}
Inserting this into equation \ref{eq:Kappa4SupIntEstimateIntermediate} yields
\begin{equation}
\begin{aligned}
&\sup_{p \in \mathbb{R}^d} \int_{\mathbb{R}^d} \left \vert \partial_k^l \kappa_4 \left( p, q, -q \right) \right \vert \mathrm{d} q \\
&\le
\frac{s_{d-1}}{2 d} \lambda
\sum_{a = 0}^{l}
\sum_{b = 0}^a
\binom{a}{b} b! \left \vert \gamma^l_a \right \vert
\left( l + a - b \right)!
\frac{\left \vert m \right \vert^3 k}{k^l + m^l}
\\
&\phantom{\le} + \frac{s_{d-1}}{2 d} \lambda
\sum_{a = 0}^{l}
\sum_{b = 0}^a
\sum_{c = 1}^{a-b}
\binom{a}{b} \binom{a - b}{c} b! \left( \frac{c}{e} \right)^c
\left \vert \gamma^l_a \right \vert \\
&\phantom{\le + \frac{s_{d-1}}{2 d}
\sum_{a = 0}^{l}
\sum_{b = 0}^a}
\times
\left( l + a - b - c \right)!
\frac{\left \vert m \right \vert^3 k}{k^l + m^l} \, .
\end{aligned}
\end{equation}
\end{proof}
Finally, the relevant estimates for $\kappa_2$ can be proved:
\begin{theorem}
Let $n \in \mathbb{N}$.
Then, there exists a constant $B_2^n \ge 0$ such that
\begin{equation}
\left \Vert \partial_k^n \kappa_2 \right \Vert_{L^\infty} \le B_2^n \frac{m^2}{k^n} \, .
\end{equation}
\end{theorem}
\begin{proof}
Let us first consider the case $n = 1$:
\begin{equation}
\label{eq:KDerivativeKappa2Estimate}
\begin{aligned}
\Vert \partial_k &\kappa_2 \Vert_{L^\infty} \\
&\le
\frac{\left( 2 \pi \right)^{-d}}{2} \left \Vert K \right \Vert_{L^\infty}
\sup_{p \in \mathbb{R}^d}
\int_{\mathbb{R}^d} \left \vert \kappa_4 \left( p, q, -q \right) \right \vert \mathrm{d} q \\
&\le
\frac{\left( 2 \pi \right)^{-d}}{2} R_1 A_4^0 \frac{\left \vert m \right \vert^3 k^2}{\left( k^2 + m^2 \right)^2} \\
&\le
\frac{\left( 2 \pi \right)^{-d}}{4} R_1 A_4^0 \frac{m^2}{k} \, ,
\end{aligned}
\end{equation}
where the second inequality follows from corollary \ref{cor:KSupEstimate} and theorem \ref{thm:Kappa4SupIntEstimate}.
Let us now proceed by induction.
Fix some $n \in \mathbb{N}$ and assume that the theorem holds for all $l \in \mathbb{N}_{\le n}$.
Then
\begin{equation}
\label{eq:Kappa2DerivativeSupEstimateIntermediate}
\begin{aligned}
\big \Vert &\partial_k^{n+1} \kappa_2 \big \Vert_{L^\infty} \\
&\le
\frac{\left( 2 \pi \right)^{-d}}{2}
\sum_{l = 0}^n \binom{n}{l}
\left \Vert \partial_k^{l} K \right \Vert_{L^\infty} \\
&\phantom{\frac{\left( 2 \pi \right)^{-d}}{2}} \times
\sup_{p \in \mathbb{R}^d} \int_{\mathbb{R}^d} \left \vert
\partial_k^{n - l} \kappa_4 \left( p, q, -q \right) \right \vert
\mathrm{d} q \\
&\le
\frac{\left( 2 \pi \right)^{-d}}{2}
\sum_{l = 0}^n \binom{n}{l} \frac{A_4^{n-l} \left \vert m \right \vert^3 k}{k^{n-l} + \left \vert m \right \vert^{n-l}}
\left \Vert \partial_k^{l} K \right \Vert_{L^\infty} \, .
\end{aligned}
\end{equation}
But, also
\begin{equation}
\begin{aligned}
&\left \Vert \partial_k^{l} K \right \Vert_{L^\infty} \\
&\le
\sum_{a = 0}^{l} \sum_{b = 0}^a
\binom{l}{a} \binom{a}{b}
\left \Vert \partial_k^{1+l-a} \bar{r} \right \Vert_{L^\infty} \\
&\phantom{\le} \times
\left \Vert \partial_k^{a-b} \left( \kappa_2 + \bar{r} \right)^{-1} \right \Vert_{L^\infty}
\left \Vert \partial_k^{b} \left( \kappa_2 + \bar{r} \right)^{-1} \right \Vert_{L^\infty} \\
&\le
\sum_{a = 0}^{l} \sum_{b = 0}^a
\binom{l}{a} \binom{a}{b}
R_{1+l-a} k^{1+a-l} \\
&\phantom{\le} \times
\left \Vert \partial_k^{a-b} \left( \kappa_2 + \bar{r} \right)^{-1} \right \Vert_{L^\infty}
\left \Vert \partial_k^{b} \left( \kappa_2 + \bar{r} \right)^{-1} \right \Vert_{L^\infty} \, .
\end{aligned}
\end{equation}
While equation \ref{eq:PropagatorNthDerivative} has been derived in a non-commutative algebra of operators, it also holds in a similar form in the commutative algebra of functions.
Thus, together with equation \ref{eq:inversePropagatorEstimate} and the induction hypothesis, we obtain
\begin{equation}
\begin{aligned}
&\left \Vert \partial_k^l \left( \kappa_2 + \bar{r} \right)^{-1} \right \Vert_{L^\infty} \\
&\le
\sum_{c \in \CombinationSet{l}} \frac{l!}{c!}
\frac{1}{m^2 + k^2}
\prod_{a = 1}^{\# c}
\frac{\left \Vert \partial_k^{c_a} \kappa_2 \right \Vert_{L^\infty} + \left \Vert \partial_k^{c_a} \bar{r} \right \Vert_{L^\infty}}{m^2 + k^2} \\
&\le
\sum_{c \in \CombinationSet{l}} \frac{l!}{c!}
\frac{1}{m^2 + k^2}
\prod_{a = 1}^{\# c}
\left( B_2^{c_a} + R_{c_a} \right) k^{-c_a} \\
&=
\frac{k^{-l}}{m^2 + k^2}
\sum_{c \in \CombinationSet{l}} \frac{l!}{c!}
\prod_{a = 1}^{\# c}
\left( B_2^{c_a} + R_{c_a} \right) \\
&=: B_{2,r}^l \frac{k^{-l}}{m^2 + k^2}
\end{aligned}
\end{equation}
for all $l \in \mathbb{N}_{\le n} \cup \left \{ 0 \right \}$, where we set $B^0_{2,r} = 1$.
Inserted into the previous equation, this gives us
\begin{equation}
\begin{aligned}
\left \Vert \partial_k^{l} K \right \Vert_{L^\infty}
&\le
\sum_{a = 0}^{l} \sum_{b = 0}^a
\binom{l}{a} \binom{a}{b}
R_{1+l-a} k^{1+a-l} \\
&\phantom{\sum_{a = 0}^{l} \sum_{b = 0}^a} \times
B_{2,r}^{a-b} \frac{k^{b-a}}{m^2 + k^2}
B_{2,r}^{b} \frac{k^{-b}}{m^2 + k^2} \\
&=:
B_K^l \frac{k^{1 - l}}{ \left( m^2 + k^2 \right)^2}
\end{aligned}
\end{equation}
for all $l \in \mathbb{N}_{\le n} \cup \left \{ 0 \right \}$.
Finally, inserting this into equation \ref{eq:Kappa2DerivativeSupEstimateIntermediate} leads to
\begin{equation}
\begin{aligned}
\left \Vert \partial_k^{n+1} \kappa_2 \right \Vert_{L^\infty}
&\le
\frac{\left( 2 \pi \right)^{-d}}{2}
\sum_{l = 0}^n \binom{n}{l} A_4^{n-l} B_K^l \\
&\phantom{\le} \times
\frac{\left \vert m \right \vert^3 k^{2 - l}}{ \left( m^2 + k^2 \right)^2 \left( k^{n-l} + \left \vert m \right \vert^{n-l} \right)} \\
&\le
\frac{\left( 2 \pi \right)^{-d}}{4}
\sum_{l = 0}^n \binom{n}{l} A_4^{n-l} B_K^l \frac{m^2}{k^{n+1}} \, .
\end{aligned}
\end{equation}
\end{proof}
From this proof, we also obtain the following extremely useful corollaries:
\begin{corollary}
\label{cor:InverseRegularizedPropagatorKDerivativeEstimate}
For all $l \in \mathbb{N}_0$, there is a constant $B_{2,r}^l \ge 0$ such that
\begin{equation}
\left \Vert \partial_k^l \left( \kappa_2 + \bar{r} \right)^{-l} \right \Vert_{L^\infty} \le B_{2,r}^l \frac{k^{-l}}{m^2 + k^2} \, .
\end{equation}
\end{corollary}
\begin{corollary}
For all $l \in \mathbb{N}_0$, there is a constant $B_K^l \ge 0$ such that
\begin{equation}
\left \Vert \partial_k^l K \right \Vert_{L^\infty} \le B_K^l \frac{k^{1-l}}{\left( m^2 + k^2 \right)^2} \, .
\end{equation}
\end{corollary}
Having obtained these estimates concerning $\kappa_2$, only a few estimates regarding the exponential regulator are needed before turning to $\rho_{2n}$.
As a start, we have the following theorem.
\begin{theorem}
For all natural numbers $l \in \mathbb{N}_0$, there exist constants $\bar{R}_l \ge 0$ such that
\begin{equation}
\left \Vert \partial_k^l \bar{r} \right \Vert_{L^1}
\le
\bar{R}_l k^{2+d-l} \, .
\end{equation}
\end{theorem}
\begin{proof}
The use of equation \ref{eq:RegulatorKDerivative} yields
\begin{equation}
\begin{aligned}
\left \Vert \partial_k^l \bar{r} \right \Vert_{L^1}
&\le
\sum_{a = 1}^l \sum_{b = 0}^{l + 1 - a} \left \vert \beta^l_{a,b} \right \vert k^{2-l-2a-2b} \\
&\phantom{\le} \times
\int_{\mathbb{R}^d} \bar{r} \left( q \right)^a \left[ \left \Vert q \right \Vert^2 + \bar{r} \left( q \right) \right]^b \mathrm{d} q \\
&=
s_{d-1}
\sum_{a = 1}^l \sum_{b = 0}^{l + 1 - a}
\left \vert \beta^l_{a,b} \right \vert k^{2+d-l} \\
&\phantom{\le} \times
\int_0^{\infty}
\frac{t^{2a + 2b + d - 1} \exp \left[ t^2 \right]^b}{\left( \exp \left[ t^2 \right] - 1 \right)^{a+b}}
\mathrm{d} t \, ,
\end{aligned}
\end{equation}
where the integral is finite since $a \ge 1$ and $b \ge 0$.
\end{proof}
\begin{corollary}
\label{cor:KDerivativeEstimates}
For all natural numbers $n \in \mathbb{N}_0$, there exist constants $C_K^n > 0$ such that
\begin{equation}
\left \Vert \partial_k^n K \right \Vert_{L^1}
\le
C_K^n \frac{k^{d + 1 - n}}{\left( m^2 + k^2 \right)^2} \, .
\end{equation}
\end{corollary}
\begin{proof}
We obviously have
\begin{equation}
\begin{aligned}
&\left \Vert \partial_k^n K \right \Vert_{L^1} \\
&\le
\sum_{l = 0}^n \sum_{a = 0}^l
\binom{n}{l} \binom{l}{a} 
\left \Vert \partial_k^{1 + n - l} \bar{r} \right \Vert_{L^1} \\
&\phantom{\le} \times 
\left \Vert \partial_k^{l-a} \left( \kappa_2 + \bar{r} \right)^{-1} \right \Vert_{L^\infty}
\left \Vert \partial_k^{a} \left( \kappa_2 + \bar{r} \right)^{-1} \right \Vert_{L^\infty} \\
&\le
\sum_{l = 0}^n \sum_{a = 0}^l
B_{2,r}^{l-a} B_{2,r}^{a} \bar{R}_{1+n-l} 
\frac{k^{d+1-n}}{\left( m^2 + k^2 \right)^2} \, .
\end{aligned}
\end{equation}
\end{proof}
\begin{theorem}
\label{thm:KL1inverseDerivativeEstimates}
For all natural numbers $n \in \mathbb{N}_0$, there exist constants $\bar{C}_K^n \ge 0$ such that
\begin{equation}
\left \vert \partial_k^n \left \Vert K \right \Vert_{L^1}^{-1} \right \vert
\le
\bar{C}_K^n \frac{\left( k^2 + m^2 \right)^2}{k^{d+1}} k^{-n} \, .
\end{equation}
\end{theorem}
\begin{proof}
We have
\begin{equation}
\begin{aligned}
\Vert &K \Vert_{L^1} \\
&= \int_{\mathbb{R}^d} \frac{\partial_k \bar{r} \left( q \right)}{\left[ \kappa_2 \left( q \right ) + \bar{r} \left( q \right) \right]^2} \mathrm{d} q\\
&\ge 
\int_{\mathbb{R}^d} \frac{\partial_k \bar{r} \left( q \right)}{\left[ \left \Vert q \right \Vert^2 + \left( 1 + t_d \lambda \right) m^2 + k^2 \right]^2} \mathrm{d} q \\
&= \frac{s_{d-1}}{k^{3-d}} \int_0^\infty \frac{t^{d+3} \left[ \cosh \left( t^2 \right) -1 \right]^{-1}}{\left[ t^2 + 1 + \left( 1 + t_d \lambda \right) \lambda \frac{m^2}{k^2} \right]^2} \mathrm{d} t \\
&\ge \frac{s_{d-1} k^{d-3}}{\left[ 2 + \left( 1 + t_d \lambda \right) \frac{m^2}{k^2} \right]^2} \int_0^1 \frac{t^{d+3}}{\left[ \cosh \left( t^2 \right) -1 \right]} \mathrm{d} t \\
&\ge \frac{s_{d-1}}{\max \left \{ 2, 1 + t_d \lambda \right \}^2} X_d
\frac{k^{d+1}}{\left[ k^2 + m^2 \right]^2}
\end{aligned}
\end{equation}
with $X_d > 0$ being the numerical value of the integral which is finite for $d \ge 1$.
Thus, by inverting both sides the theorem is true for $n = 0$.
For $n \in \mathbb{N}$ note that $\left \vert \partial_k^n \left \Vert K \right \Vert_{L^1} \right \vert \le \left \Vert \partial_k^n K \right \Vert_{L^1}$ since $K \ge 0$.
Thus,
\begin{equation}
\begin{aligned}
\Big \vert \partial_k^n &\left \Vert K \right \Vert_{L^1}^{-1} \Big \vert \\
&\le
\sum_{c \in \CombinationSet{n}} \frac{n!}{c!} \left \Vert K \right \Vert_{L^1}^{-1}
\prod_{l = 1}^{\# c}
\frac{\left \Vert \partial_k^{c_l} K \right \Vert_{L^1}}{\left \Vert K \right \Vert_{L^1}} \\
&\le
\sum_{c \in \CombinationSet{n}} \frac{n!}{c!}
\bar{C}_K^0 \frac{\left( k^2 + m^2 \right)^2}{k^{d+1}}
\prod_{l = 1}^{\# c}
\bar{C}_k^0 C_K^{c_l} k^{-c_l} \\
&=
\frac{\left( k^2 + m^2 \right)^2}{k^{d+1+n}}
\sum_{c \in \CombinationSet{n}} \frac{n!}{c!}
\left( \bar{C}_K^0 \right)^{\# c + 1}
\prod_{l = 1}^{\# c}
C_K^{c_l} \, .
\end{aligned}
\end{equation}
\end{proof}
Now, we finally turn to our estimates of the higher correlation functions.
\begin{theorem}
\label{thm:Kappa2nDerivativesEstimate}
Define $\Delta$ as in equation \ref{eq:DeltaDefinition}.
Then, for all $n\in \mathbb{N}_{\ge 2}$, $x \in \mathbb{N}$ and $l \in \mathbb{N}_0$ there exist constants $B_{2n}^{l, x} \ge 0$ such that
\begin{equation}
\label{eq:BigEstimateGeneralForm}
\left \Vert \partial_k^l \kappa_{2n} \right \Vert_{L^\infty}
\le
B_{2n}^{l, x}
\frac{\left \vert m \right \vert^{d + \left(2 - d\right) n + \left( n - 2 \right) \left(1 + \Delta \right)} k^x}{\left( k + \left \vert m \right \vert \right)^{\left( n - 2 \right) \left( 1 + \Delta \right) + x + l}} \, .
\end{equation}
\end{theorem}
\begin{proof}
Let us begin with a proof of the statement for $\kappa_4$ i.e for $n = 2$.
We know from equation \ref{eq:Kappa4KDerivative}, that
\begin{equation}
\begin{aligned}
\Vert \partial_k^l \kappa_4 &\Vert_{L^\infty}
\le
\lambda \sum_{a = 0}^l \left \vert \gamma^l_a \right \vert
\frac{\left \vert m \right \vert^{4 - d + a - a d}}{k^{a + l}}
\exp \left[ - \frac{ \left \vert m \right \vert}{k} \right] \\
&\times
\sup_{y \in \mathbb{R}}
\left( y^2 + \left \vert m \right \vert^d \right)^a
\exp \left[ - \frac{y^2}{k \left \vert m \right \vert^{d-1}} \right] \, .
\end{aligned}
\end{equation}
This expands to
\begin{equation}
\begin{aligned}
\Vert \partial_k^l \kappa_4 &\Vert_{L^\infty}
\le
\lambda \sum_{a = 0}^l \sum_{b = 0}^a \binom{a}{b}
\left \vert \gamma^l_a \right \vert
\frac{\left \vert m \right \vert^{4 - d + a - b d}}{k^{a + l}} \\
&\times
\exp \left[ - \frac{ \left \vert m \right \vert}{k} \right]
\sup_{y \in \mathbb{R}}
y^{2b}
\exp \left[ - \frac{y^2}{k \left \vert m \right \vert^{d-1}} \right] \, ,
\end{aligned}
\end{equation}
such that
\begin{equation}
\begin{aligned}
\Vert \partial_k^l \kappa_4 &\Vert_{L^\infty}
\le
\lambda \exp \left[ - \frac{ \left \vert m \right \vert}{k} \right]
\sum_{a = 0}^l \Bigg( \left \vert \gamma^l_a \right \vert
\frac{\left \vert m \right \vert^{4 - d + a}}{k^{a + l}} \\
&+ \lambda \sum_{b = 1}^a \binom{a}{b} \left( \frac{b}{e} \right)^b \left \vert \gamma^l_a \right \vert
\frac{\left \vert m \right \vert^{4 - d + a - b}}{k^{a + l - b}}
\Bigg) \, .
\end{aligned}
\end{equation}
But from equation \ref{eq:ExponentialScreeningTrick} we have for all $a, b \in \mathbb{N}_0$ with $a \ge b$ and all $x \in \mathbb{N}$
\begin{equation}
\exp \left[ - \frac{ \left \vert m \right \vert}{k} \right]
\le
\left( x + l + a - b \right)!
\frac{k^{x + l + a - b} \left \vert m \right \vert^{b - a}}{k^{x + l} + \left \vert m \right \vert^{x + l}} \, .
\end{equation}
Inserted into the previous equation, this yields
\begin{equation}
\begin{aligned}
\Vert &\partial_k^l \kappa_4 \Vert_{L^\infty}
\le
\lambda \sum_{a = 0}^l \left \vert \gamma^l_a \right \vert \Bigg( 
\frac{\left( x + l + a \right)! \, k^x}{k^{x + l} + \left \vert m \right \vert^{x + l}} \left \vert m \right \vert^{4 - d} \\
&+ \sum_{b = 1}^a \binom{a}{b} \left( \frac{b}{e} \right)^b
\frac{\left( x + l + a - b \right)!}{k^{x + l} + \left \vert m \right \vert^{x + l}} k^x \left \vert m \right \vert^{4 - d}
\Bigg) \, .
\end{aligned}
\end{equation}
The result then follows from
\begin{equation}
\frac{1}{k^{x+l} +\left \vert m \right \vert^{x+l}}
\le
2^{x+l-1}
\frac{1}{\left( k + \left \vert m \right \vert \right)^{x+l}} \, .
\end{equation}
Let us now fix some $n \in \mathbb{N}_{\ge 2}$ and assume the theorem to be true for all $l \in \mathbb{N}_{\ge 2}$ with $l \le n$.
It needs to be shown that the theorem also holds for $\kappa_{2n+2}$ as given by equation \ref{eq:Kappa2np2FromRightInverse}.
By the linearity of $\rho_{2n}$ it suffices to show this for $\rho_{2n} \partial_k \kappa_{2n}$ and $\rho_{2n} \bar{\lambda}_c$ separately for all $c \in \EvenCombinationSet{2n} \setminus \left \{ \left( 2n \right) \right \}$.
In either case, for $l \in \mathbb{N}_0$ and a sufficiently regular $\SymmetricGroupN{2n-1}^\ast$-symmetric function $g$ we have
\begin{equation}
\label{eq:BigEstimateIntermediate}
\begin{aligned}
&\left \Vert \partial_k^l \rho_{2n} g \right \Vert_{L^\infty} \\
&\le
\sum_{J \subseteq \left \{ 0, ..., 2n + 1 \right \}}
\sum_{l = 0}^{\left \lfloor \frac{2n - 1 - \# J}{2} \right \rfloor}
\sum_{a = 0}^l \sum_{b = 0}^a
\binom{l}{a} \binom{a}{b} \\
&\phantom{\le} \times
\left \vert \alpha^{2n}_{\# J, l} \right \vert
\left \Vert \partial_k^{l-a} g \right \Vert_{L^\infty} 
\left \Vert \partial_k^{a-b} K^{\otimes 2n - 1 - \# J - l} \right \Vert_{L^1} \\
&\phantom{\le} \times
\left \vert \partial_k^b \left \Vert K \right \Vert_{L^1}^{- \left( 2n - \# J - l \right)} \right \vert \, ,
\end{aligned}
\end{equation}
where we have used that $\int_{\mathbb{R}^d} K = \left \Vert K \right \Vert_{L^1}$ since $K > 0$.
Employing corollary \ref{cor:KDerivativeEstimates}, we get
\begin{equation}
\begin{aligned}
\big \Vert \partial_k^a K^{\otimes b} \big \Vert_{L^1}
&\le
\sum_{\substack{\alpha \in \mathbb{N}_0^c \\ \left \vert \alpha \right \vert = a}}
\frac{a!}{\alpha!}
\prod_{j = 1}^b
\left \Vert \partial_k^{\alpha_j} K \right \Vert_{L^1} \\
&\le
\sum_{\substack{\alpha \in \mathbb{N}_0^b \\ \left \vert \alpha \right \vert = a}}
\frac{a !}{\alpha!}
\prod_{j = 1}^b
C_K^{\alpha_j} \frac{k^{d+1-\alpha_j}}{\left( m^2 + k^2 \right)^2} \\
&=:
D_K^{a, c} \left( \frac{k^{d+1}}{\left( m^2 + k^2 \right)^2} \right)^b k^{-a}
\end{aligned}
\end{equation}
for all $a, b \in \mathbb{N}_0$\footnote{The case $b = 0$ is trivial as $K^{\otimes 0} = 1$.}.
Furthermore, from theorem \ref{thm:KL1inverseDerivativeEstimates} one has
\begin{equation}
\begin{aligned}
\left \vert \partial_k^a \left \Vert K \right \Vert_{L^1}^{-b} \right \vert
&\le
\sum_{\substack{\alpha \in \mathbb{N}_0^{b} \\ \left \vert \alpha \right \vert = a}}
\frac{a!}{\alpha!}
\prod_{j = 1}^c
\left \vert \partial_k^{\alpha_j} \left \Vert K \right \Vert_{L^1}^{-1} \right \vert \\
&\le
\sum_{\substack{\alpha \in \mathbb{N}_0^b \\ \left \vert \alpha \right \vert = a}}
\frac{a!}{\alpha!}
\prod_{j = 1}^b
\bar{C}_K^{\alpha_j} \frac{\left( k^2 + m^2 \right)^2}{k^{d+1}} k^{-\alpha_j} \\
&=:
\bar{D}_K^{a, b} \left( \frac{\left( m^2 + k^2 \right)^2}{k^{d+1}} \right)^{b} k^{-a}
\end{aligned}
\end{equation}
for all $a \in \mathbb{N}_0$ and $b \in \mathbb{N}$.
The insertion of these two inequalities into equation \ref{eq:BigEstimateIntermediate} reveals the important intermediate result
\begin{equation}
\label{eq:Rho2nGKDerivativeEstimate}
\begin{aligned}
&\left \Vert \partial_k^l \rho_{2n} g \right \Vert_{L^\infty} \\
&\le
\sum_{J \subseteq \left \{ 1, ..., 2n + 2 \right \}}
\sum_{j = 0}^{\left \lfloor \frac{2n - 1 - \# J}{2} \right \rfloor}
\sum_{a = 0}^l \sum_{b = 0}^a \\
&\phantom{\le}
\binom{l}{a} \binom{a}{b}
\left \vert \alpha^{2n}_{\# J, j} \right \vert \\
&\phantom{\le} \times
D_K^{a-b,2 n - 1 - \# J - j}
\bar{D}_K^{b,2 n - \# J - j} \\
&\phantom{\le} \times
\frac{\left( m^2 + k^2 \right)^2}{k^{d + 1 + a}}
\left \Vert \partial_k^{l-a} g \right \Vert_{L^\infty} \\
&:=
\sum_{a = 0}^l E_{2n}^{l,a}
\frac{\left( m^2 + k^2 \right)^2}{k^{d + 1 + a}}
\left \Vert \partial_k^{l-a} g \right \Vert_{L^\infty} \, .
\end{aligned}
\end{equation}
The divergent behaviour for $k \to 0$ elucidates the need for the extremely strong IR regularity of $\kappa_4$ as imposed in equation \ref{eq:Kappa4Definition}.
Now, consider the case $g = \partial_k \kappa_{2n}$ and $x \in \mathbb{N}$:
\begin{equation}
\begin{aligned}
&\left \Vert \partial_k^l \rho_{2n} \partial_k \kappa_{2n} \right \Vert_{L^\infty} \\
&\le
\sum_{a = 0}^l E_{2n}^{l,a}
\frac{\left( m^2 + k^2 \right)^2}{k^{d + 1 + a}}
\left \Vert \partial_k^{l+1-a} \kappa_{2n} \right \Vert_{L^\infty} \\
&\le
\sum_{a = 0}^l E_{2n}^{l,a} B_{2n}^{l+1-a, x+d+1+a}
\frac{\left( m^2 + k^2 \right)^2}{k^{d + 1 + a}} \\
&\phantom{\le \sum_{a = 0}^l} \times
\frac{\left \vert m \right \vert^{d + \left(2-d \right) n + \left( n - 2 \right)\left(1 + \Delta \right)} k^{x+d+1+a}}{\left(k + \left \vert m \right \vert \right)^{\left(n-2\right)\left(1 + \Delta\right) + l + 2 + d + x}} \\
&=
\sum_{a = 0}^l E_{2n}^{l,a} B_{2n}^{l+1-a, x+d+1+a} \\
&\phantom{\le \sum_{a = 0}^l} \times
\frac{\left( m^2 + k^2 \right)^2 \left \vert m \right \vert^{d - 3 - \Delta}}{\left( k + \left \vert m \right \vert \right)^{d + 1 - \Delta}} \\
&\phantom{\le \sum_{a = 0}^l} \times
\frac{\left \vert m \right \vert^{d + \left(2 - d\right)\left(n + 1 \right) + \left( n - 1 \right) \left(1 + \Delta \right)} k^{x}}{\left(k + \left \vert m \right \vert \right)^{\left(n -1\right)\left( 1 + \Delta \right) + l + x}} \\
&\le
\sum_{a = 0}^l E_{2n}^{l,a} B_{2n}^{l+1-a, x+d+1+a} \\
&\phantom{\le \sum_{a = 0}^l} \times
\frac{\left \vert m \right \vert^{d + \left(2 - d\right)\left(n + 1 \right) + \left( n - 1 \right) \left(1 + \Delta \right)} k^{x}}{\left(k + \left \vert m \right \vert \right)^{\left(n -1\right)\left( 1 + \Delta \right) + l + x}} \, .
\end{aligned}
\end{equation}
This is the expected result and also shows that the use of these methods requires $d - 3 - \Delta \ge 0$.
Otherwise, the last inequality would not generally hold.
It remains to estimate the $\rho_{2n} \bar{\lambda}_c$ terms.
Let $c \in \EvenCombinationSet{2n} \setminus \left \{ \left( 2n \right) \right \}$.
Then, obviously
\begin{equation}
\left \Vert \partial_k^l \bar{\lambda}_c \right \Vert_{L^\infty} \le \left \Vert \partial_k^l \lambda_c \right \Vert_{L^\infty} \, ,
\end{equation}
so that we need not bother with symmetrisation.
In total, for $g = \bar{\lambda}_c$ and $l \in \mathbb{N}_0$ we get
\begin{equation}
\label{eq:Rho2nLambdaCKDerivativeEstimate1}
\begin{aligned}
\big \Vert \partial_k^l \rho_{2n} \bar{\lambda}_c &\big \Vert_{L^\infty} \\
&\le
\sum_{a = 0}^l E_{2n}^{l,a}
\frac{\left( m^2 + k^2 \right)^2}{k^{d + 1 + a}}
\left \Vert \partial_k^{l-a} \lambda_c \right \Vert_{L^\infty} \, .
\end{aligned}
\end{equation}
Estimating $\left \Vert \partial_k^{l-a} \lambda_c \right \Vert_{L^\infty}$ is rather cumbersome with
\begin{equation}
\begin{aligned}
\left \Vert \partial_k^l \lambda_c \right \Vert_{L^\infty}
&\le
\sum_{a = 0}^{l} \sum_{b = 0}^{a}
\binom{l}{a} \binom{a}{b}
\left \Vert \partial_k^{l-a} K \right \Vert_{L^1} \\
&\phantom{\le} \times
\left \Vert \partial_k^{a-b} \left( \left[ \kappa_2 + r \right]^{-1} \right)^{\otimes \# c - 1} \right \Vert_{L^\infty} \\
&\phantom{\le} \times
\left \Vert \partial_k^b
\bigotimes_{j = 1}^{\# c} \kappa_{2 + c_j} \right \Vert_{L^\infty}
\end{aligned}
\end{equation}
for all $l \in \mathbb{N}_0$.
However, using corollary \ref{cor:KDerivativeEstimates} and \ref{cor:InverseRegularizedPropagatorKDerivativeEstimate} one obtains
\begin{equation}
\begin{aligned}
&\left \Vert \partial_k^l \lambda_c \right \Vert_{L^\infty} \\
&\le
\sum_{a = 0}^{l} \sum_{b = 0}^{a}
\binom{l}{a} \binom{a}{b}
C_K^{l - a} \frac{k^{d + 1 - l + a}}{\left( m^2 + k^2 \right)^2} \\
&\phantom{\le} \times
\sum_{\substack{\alpha \in \mathbb{N}_0^{\# c - 1} \\ \left \vert \alpha \right \vert = a - b}}
\frac{\left(a - b \right)!}{\alpha!}
\prod_{j = 1}^{\# c - 1} B_{2,r}^{\alpha_j} \frac{k^{- \alpha_j}}{m^2 + k^2} \\
&\phantom{\le} \times \left(
\sum_{\substack{\beta \in \mathbb{N}_0^{\# c} \\ \left \vert \beta \right \vert = b}}
\frac{b!}{\beta!}
\prod_{i = 1}^{\# c}
\left \Vert \partial_k^{\beta_i} \kappa_{2 + c_i} \right \Vert_{L^\infty}
\right) \\
&=
\sum_{a = 0}^{l} \sum_{b = 0}^{a}
\binom{l}{a} \binom{a}{b}
C_K^{l - a} \frac{k^{d + 1 - l + b}}{\left( m^2 + k^2 \right)^{\# c + 1}} \\
&\phantom{\le} \times
\sum_{\substack{\alpha \in \mathbb{N}_0^{\# c - 1} \\ \left \vert \alpha \right \vert = a - b}}
\frac{\left(a - b \right)!}{\alpha!} B_{2,r}^{\alpha_j} \\
&\phantom{\le} \times \left(
\sum_{\substack{\beta \in \mathbb{N}_0^{\# c} \\ \left \vert \beta \right \vert = b}}
\frac{b!}{\beta!}
\prod_{i = 1}^{\# c}
\left \Vert \partial_k^{\beta_i} \kappa_{2 + c_i} \right \Vert_{L^\infty}
\right) \\
&:=
\sum_{\substack{\beta \in \mathbb{N}_0^{\# c} \\ \left \vert \beta \right \vert \le l}}
\frac{F_c^{l,\beta} k^{d + 1 - l + \left \vert \beta \right \vert}}{\left( m^2 + k^2 \right)^{\# c + 1}}
\prod_{j = 1}^{\# c}
\left \Vert \partial_k^{\beta_j} \kappa_{2 + c_j} \right \Vert_{L^\infty} \, .
\end{aligned}
\end{equation}
Inserting this result into equation \ref{eq:Rho2nLambdaCKDerivativeEstimate1} yields
\begin{equation}
\label{eq:Rho2nLambdaCKDerivativeEstimate2}
\begin{aligned}
&\left \Vert \partial_k^l \rho_{2n} \bar{\lambda}_c \right \Vert_{L^\infty} \\
&\le
\sum_{a = 0}^l \sum_{\substack{\beta \in \mathbb{N}_0^{\# c} \\ \left \vert \beta \right \vert \le l - a}}
E_{2n}^{l,a} F_c^{l-a,\beta} \\
&\phantom{\le} \times
\frac{k^{\left \vert \beta \right \vert - l}}{\left( m^2 + k^2 \right)^{\# c - 1}}
\prod_{j = 1}^{\# c}
\left \Vert \partial_k^{\beta_j} \kappa_{2 + c_j} \right \Vert_{L^\infty} \\
&:=
\sum_{\substack{\beta \in \mathbb{N}_0^{\# c} \\ \left \vert \beta \right \vert \le l}}
\frac{G_c^{l,\beta} k^{\left \vert \beta \right \vert - l}}{\left( m^2 + k^2 \right)^{\# c - 1}}
\prod_{j = 1}^{\# c}
\left \Vert \partial_k^{\beta_j} \kappa_{2 + c_j} \right \Vert_{L^\infty} \, ,
\end{aligned}
\end{equation}
so that we just need to use proper estimates for $\left \Vert \partial_k^{\beta_n} \kappa_{2 + c_n} \right \Vert_{L^\infty}$.
To that end, let $x \in \mathbb{N}$ be arbitrary and fix some multi-index $X \in \mathbb{N}_0^{\# c}$ with $\left \vert X \right \vert = x + l$.
Invoking the induction hypothesis leads to the conclusion
\begin{equation}
\begin{aligned}
\Big \Vert& \partial_k^{\beta_j} \kappa_{2 + c_j} \Big \Vert \\
&\le
B_{2+c_j}^{\beta_j,X_j}
\frac{\left \vert m \right \vert^{d + \left( 2 - d \right) \left( \frac{c_j}{2} + 1 \right) + \left( \frac{c_j}{2} - 1 \right) \left( 1 + \Delta \right)} k^{X_j}}{\left( k + \left \vert m \right \vert \right)^{\left( \frac{c_j}{2} - 1 \right) \left( 1 + \Delta \right) + X_j + \beta_j}}
\end{aligned}
\end{equation}
for all even $c_j \in \mathbb{N}_{\le 2n-2}$ and all $\beta_j \in \mathbb{N}_0$.
In particular, this translates to
\begin{equation}
\begin{aligned}
&\prod_{j = 1}^{\# c}
\left \Vert \partial_k^{\beta_j} \kappa_{2 + c_j} \right \Vert_{L^\infty} \\
&\le
\frac{\left \vert m \right \vert^{\left(1 - \Delta\right) \# c + \left( 3 + \Delta - d \right) n} k^{x + l}}{\left( k + \left \vert m \right \vert \right)^{\left( 1 + \Delta \right) \left( n - \# c \right) + x + b + l}}
\prod_{j = 1}^{\# c} B_{2+c_j}^{\beta_j,X_j} \\
&=
\frac{\left \vert m \right \vert^{\left(1 - \Delta\right) \left( \# c - 1 \right)}}{\left( k + \left \vert m \right \vert \right)^{\left(1 - \Delta\right) \left( \# c - 1 \right)}} \\
&\phantom{=} \times
\frac{\left \vert m \right \vert^{\left( 3 + \Delta - d \right) n + 1 - \Delta} k^{x + l}}{\left( k + \left \vert m \right \vert \right)^{n + \Delta n - 2 \# c + x + b + l + 1 - \Delta}}
\prod_{j = 1}^{\# c} B_{2+c_j}^{\beta_j,X_j} \\
&\le
\frac{\left \vert m \right \vert^{\left( 3 + \Delta - d \right) n + 1 - \Delta} k^{x + l}}{\left( k + \left \vert m \right \vert \right)^{n + \Delta n - 2 \# c + x + b + l + 1 - \Delta}}
\prod_{j = 1}^{\# c} B_{2+c_j}^{\beta_j,X_j} \\
&\le
\frac{\left \vert m \right \vert^{\left( 3 + \Delta - d \right) n + 1 - \Delta} k^{x + l} \left( k^2 + m^2 \right)^{\# c - 1}}{\left( k + \left \vert m \right \vert \right)^{n + \Delta n + x + l - 1 - \Delta} k^b} \\
&\phantom{\le} \times
2^{\# c - 1}
\prod_{j = 1}^{\# c} B_{2+c_j}^{\beta_j,X_j}
\end{aligned}
\end{equation}
with $\left \vert c \right \vert = 2n$ and $\left \vert \beta \right \vert = b$.
Here, it may be seen that it was important to choose $\Delta \le 1$
Otherwise, the second inequality would in general not hold.
Thus, the largest $\Delta$ that is possible using these methods is $\max \left \{ d - 3, 1 \right \}$ which precisely corresponds to the choice made in equation \ref{eq:DeltaDefinition}.
We may now insert this result into equation \ref{eq:Rho2nLambdaCKDerivativeEstimate2} obtaining
\begin{equation}
\begin{aligned}
\left \Vert \partial_k^l \rho_{2n} \bar{\lambda}_c \right \Vert_{L^\infty}
&\le
2^{\# c - 1}
\frac{\left \vert m \right \vert^{\left( 3 + \Delta - d \right) n + 1 - \Delta} k^x}{\left( k + \left \vert m \right \vert \right)^{\left( n - 1 \right) \left( \Delta + 1 \right) + x + l}} \\
&\phantom{\le} \times
\sum_{b = 0}^l \sum_{\substack{\beta \in \mathbb{N}_0^{\# c} \\ \left \vert \beta \right \vert = b}}
G_c^{l,\beta}
\prod_{j = 1}^{\# c} B_{2+c_j}^{\beta_j,X_j} \\
&:=
H_c^{l,X} \frac{\left \vert m \right \vert^{\left( 3 + \Delta - d \right) n + 1 - \Delta} k^x}{\left( k + \left \vert m \right \vert \right)^{\left( n - 1 \right) \left( \Delta + 1 \right) + x + l}} \, .
\end{aligned}
\end{equation}
Due to our appropriate choice of $X$ and our careful estimates the right-hand side precisely corresponds to the one of equation \ref{eq:BigEstimateGeneralForm} with $n$ replaced by $n + 1$.
\end{proof}

\cleardoublepage
\bibliography{bibliography}
\end{document}